\documentclass[acmsmall,screen,nonacm]{acmart}

\setcopyright{none}
\makeatletter
\@ifundefined{hyxmp@parse@acmart}{}{\let\hyxmp@parse@acmart\relax}
\makeatother

\usepackage{amsmath}
\usepackage{mathtools}
\usepackage{booktabs}
\usepackage{tabularx}
\usepackage{array}
\usepackage{multirow}
\usepackage{enumitem}
\usepackage{graphicx}
\usepackage{xcolor}
\usepackage{tikz}
\usepackage{url}
\usepackage{microtype}
\usepackage{algorithm}
\usepackage{algpseudocode}
\usepackage{placeins}
\usetikzlibrary{arrows.meta,positioning,fit,calc,shapes.geometric}

\hypersetup{
  hidelinks,
  pdftitle={Authorization Revocation for Long-Running AI Agents: Root-Scoped Quiescence under Delegation and Asynchronous Execution},
  pdfauthor={Genliang Zhu and Chu Wang},
  pdfsubject={Root-scoped authorization quiescence for long-running agents},
  pdfcreator={LaTeX}
}

\setlist[itemize]{leftmargin=1.25em}
\setlist[enumerate]{leftmargin=1.45em}
\newcolumntype{L}[1]{>{\raggedright\arraybackslash}p{#1}}
\newcolumntype{Y}{>{\raggedright\arraybackslash}X}

\definecolor{protocolBlue}{HTML}{2563EB}
\definecolor{protocolBlueFill}{HTML}{EFF6FF}
\definecolor{protocolGreen}{HTML}{16A34A}
\definecolor{protocolGreenFill}{HTML}{F0FDF4}
\definecolor{protocolAmber}{HTML}{D97706}
\definecolor{protocolAmberFill}{HTML}{FFFBEB}
\definecolor{protocolPurple}{HTML}{7C3AED}
\definecolor{protocolPurpleFill}{HTML}{F5F3FF}
\definecolor{protocolRed}{HTML}{DC2626}
\definecolor{protocolRedFill}{HTML}{FEF2F2}

\newtheorem{definition}{Definition}
\newtheorem{lemma}{Lemma}
\newtheorem{theorem}{Theorem}

\newcommand{\mcode}[1]{\ifmmode\text{\normalfont\ttfamily #1}\else{\normalfont\ttfamily #1}\fi}

\newcommand{\Q}{\textsc{Quiescent}}
\newcommand{\NQ}{\textsc{Not-Quiescent}}
\newcommand{\UQ}{\textsc{Indeterminate}}

\title[Authorization Revocation for Long-Running AI Agents]
{Authorization Revocation for Long-Running AI Agents: Root-Scoped Quiescence
under Delegation and Asynchronous Execution}

\author{Genliang Zhu}
\affiliation{%
  \institution{Accentrust}
  \city{Vancouver}
  \country{Canada}
}
\affiliation{%
  \institution{Georgia Institute of Technology}
  \city{Atlanta}
  \state{Georgia}
  \country{USA}
}
\email{research@accentrust.com}

\author{Chu Wang}
\affiliation{%
  \institution{Accentrust}
  \city{Vancouver}
  \country{Canada}
}
\affiliation{%
  \institution{University of Illinois Urbana-Champaign}
  \city{Urbana}
  \state{Illinois}
  \country{USA}
}

\begin{abstract}
Long-running AI agents outlive initiating processes through credentials,
delegated tasks, queues, callbacks, reservations, and provider-side operations.
Cancellation, process exit, and credential revocation neither close every
pre-cut carrier nor distinguish independently authorized shared work. We define
\emph{root-scoped authorization quiescence}: for each manifested sink, a
certificate accounts for every cut-relevant acceptance under the retired
root-epoch atom that precedes its local fence and excludes protected acceptance
under that atom after the fence, while permitting exact rebind to a current,
independently sufficient support.

The root-scoped quiescence protocol linearizes a root cut, fences old-root expansion and
protected sinks, represents alternative and conjunctive authority as antichains
of minimal sufficient root sets, and composes provider-frontier certificates
into a cutset over registered old-root paths. Exact channel-token accounting
reconciles transfers; missing or conflicting evidence remains indeterminate.
Under stated assumptions, we prove post-cut issuer non-expansion,
support-sound projection, compositional soundness under exact channel conservation,
independent-support preservation, merge-order independence, and crash/replay
stability.

A provider-free late-effect test suite matches 17/17 registered outcomes. Two
cancellation-only and one cut-only execution accept the same class of already
scheduled late effect; two cut-plus-fence executions, one restart, and one
stale-process execution reject it. A separately implemented checker verifies
17/17 traces and rejects 44/44 consistently rehashed semantic regressions. The
certificate establishes root-relative authorization quiescence within its
bound manifest and configuration, not global idleness, rollback, or business
completion.
\end{abstract}


\begin{document}
\maketitle
\hypersetup{
  pdftitle={Authorization Revocation for Long-Running AI Agents: Root-Scoped Quiescence under Delegation and Asynchronous Execution},
  pdfauthor={Genliang Zhu and Chu Wang},
  pdfsubject={Root-scoped authorization quiescence for long-running agents},
  pdfcreator={LaTeX}
}

\section{Introduction}
\label{sec:introduction}

Long-running agents routinely outlive the interaction that launched them. A
single task can create remote subtasks, durable queue entries, scheduled jobs,
webhooks, delegated credentials, reservations, and provider-side operations.
Stopping the local process therefore need not stop the task's authority. Even
an acknowledged cancellation can race with a remote completion, and a revoked
token can leave a prepared operation or separately delegated credential able
to commit. The operational question is not merely whether a controller sent a
stop request. It is whether the retired authority has lost every remaining path
to a future protected effect and whether every cut-relevant old-atom frontier
crossing is terminally accounted.

Existing interfaces expose observations that are useful but individually
insufficient for that conclusion. MCP task cancellation and A2A cancellation
are cooperative control transitions rather than proofs that every underlying
effect has stopped \cite{mcp2026tasks,a2a2026spec}. OAuth token revocation
invalidates a named token at the authorization server, subject to the
deployment's propagation and mediation semantics \cite{rfc7009}. Process exit
establishes local process termination. Provider-local closure can establish
that one provider has no unresolved effect frontier. None of these observations
alone accounts for a pre-cut message already accepted by another provider, a
credential derived before revocation, or a shared agent that remains valid
under an independent authorization root.

This paper studies \emph{root-scoped authorization quiescence}. Let $r$ be an
authorization root and $c$ a durable cut targeting $(r,e)$ that retires root epoch $e$ and
advances issuer state to $e+1$. A quiescence certificate for $c$ asserts that
every registered carrier and obligation whose selected support depends on the
retired atom is terminated, effect-closed, exactly accounted as pre-fence
committed, or atomically rebound to an independently sufficient witness that
excludes the retired atom; every cut-relevant acceptance under that atom that precedes its
provider-local fence is exactly accounted; every required sink has installed
its fence; and every old-atom handoff in the frozen channel frontier is
terminally accounted. No protected-sink acceptance ordered after its fence may
depend on the retired atom. The assertion is relative to one root and one cut.
It neither declares the whole system idle nor revokes work supported by other
roots.

\subsection{A strict counterexample}

Consider a shared remote worker supported by roots $r_A$ and $r_B$. Before
shutdown, work under $r_A$ creates a durable queue message and a provider-side
reservation. The controller then receives a successful cancellation response,
revokes the initiating token, and observes that the local process has exited.
The queue has already accepted the message, so a remote worker can still
convert the reservation into a protected effect. Declaring quiescence from the
three local observations is false.

Blind subtree revocation is not a correct repair. The shared worker may have an
independent, still-current witness rooted only in $r_B$. Deleting the worker
would over-revoke valid work, while retaining it without changing the selected
support could launder $r_A$'s authority through the shared principal. A sound
protocol must therefore solve two problems together: close every effect path
selected under $r_A$, including paths crossing provider boundaries, and
preserve continuing authority only through an explicit atomic rebind to a
support witness that excludes the retired root-$A$ atom.

\subsection{Protocol overview}

We present a root-scoped quiescence protocol that establishes and
certifies this root-scoped cut. The protocol maintains an outstanding
authority--obligation graph. Each live authority instance records an antichain
of minimal sufficient root-support witnesses and one currently selected
witness. Carriers represent the means by which authority can outlive the
initiator---credentials, tasks, messages, triggers, reservations, and prepared
effects. Directed channel tokens represent work transferred between adapters
or providers.

Shutdown begins with a durable epoch transition that prevents new authority
from being issued under the retired epoch. Required provider adapters install
operation-time fences, enumerate or conservatively cover their local frontier,
and drive each covered item to a terminal, effect-closed, pre-fence committed,
rebound, or unknown state. An adapter then emits a signed leaf certificate
binding the root, cut, coverage manifest, configuration, fence and log
watermarks, rebind receipts, local-open set, and outbound and accepted channel
tokens. The global combiner accepts a set of leaves only when the manifest is
complete, the leaves agree on the cut, no unknown or conflicting state remains,
and cross-provider channel tokens balance exactly. Missing evidence yields
\textsc{Indeterminate}; known live work yields \textsc{Not-Quiescent}. Neither
state is converted into success by timeout.

This structure deliberately separates actuation from evidence. Provider-local
effect-closure certificates are valuable leaves, but a collection of locally
empty leaves is not yet a global proof: a message can be in flight between two
locally empty frontiers. Conversely, global cancellation orchestration without
provider effect fences cannot exclude a late commit. The protocol composes
both layers at one root-epoch cut.

\subsection{Contributions}

This paper makes five contributions:

\begin{enumerate}
  \item It defines root-scoped authorization quiescence and gives strict
  executions in which cancellation acknowledgement, token revocation, process
  exit, subtree revocation, or provider-local closure is insufficient. The
  definition distinguishes cessation of one authorization root from global
  system idleness and from business-level completion or rollback.

  \item It introduces a support algebra for multi-root authority. Minimal
  sufficient root sets form an antichain, live instances select one witness,
  and continued use across shutdown requires an atomic rebind receipt to a
  witness excluding the retired atom. This preserves independent authority
  without permitting root laundering.

  \item It gives an active cut protocol that combines a durable root-epoch
  transition, issuance freeze, sink fencing, carrier and obligation closure,
  exact inter-adapter channel accounting, and signed provider leaf
  certificates. Leaf certificates form a cutset proof rather than an
  unordered collection of local cancellation claims.

  \item It establishes conditional safety and composition results: post-cut
  non-issuance, no false quiescence, independent-support preservation, merge-order
  independence, and crash/replay stability. It also identifies an
  indistinguishability boundary: an opaque endpoint that supplies no sound
  query, fence, expiry, or terminal receipt cannot contribute a positive
  quiescence leaf.

  \item It implements an executable late-effect test-suite evaluation and a separately
  implemented trace checker for cancellation races, late effects, crash and
  restart, missing fences, in-flight channels, alternate roots, opaque
  endpoints, and certificate replay. Section~\ref{sec:evaluation} reports only
  executions accepted by the registered verifier.
\end{enumerate}

\subsection{Scope}

The certified safety property covers registered endpoints and protected effect
sinks whose authority lineage, carriers, and cross-provider transfers are
completely mediated or conservatively represented. It is a claim that no
future protected acceptance can depend on the retired root atom and that every
cut-relevant old-atom acceptance preceding its provider-local fence is exactly
accounted.
Irreversible effects committed before those fences remain historical facts;
task goal completion, compensation, global process passivity, physical-world
cessation, and provider behavior outside the registered evidence boundary are
separate properties.

The remainder of the paper separates quiescence from adjacent control results
(Section~\ref{sec:background}), states the system and assurance model
(Section~\ref{sec:system}), formalizes the support and cut semantics
(Section~\ref{sec:formal}), presents the protocol and instantiations
(Sections~\ref{sec:design}--\ref{sec:instantiations}), evaluates the registered
artifact (Section~\ref{sec:evaluation}), and closes with related work and the
exact assurance boundary.

\section{Background and Problem Separation}
\label{sec:background}

The protocol composes established control mechanisms, but its decision
object is distinct: whether one authorization root can still cause a new
protected commit after a particular durable cut. Table~\ref{tab:separation}
states the evidence supplied by adjacent mechanisms and the remaining fact
needed for that decision.

\begin{table}[t]
\caption{Adjacent observations and the root-scoped quiescence decision.}
\label{tab:separation}
\small
\begin{tabularx}{\textwidth}{p{0.18\textwidth} X X}
\toprule
Observation or mechanism & What it establishes & What remains for root-scoped quiescence \\
\midrule
Cancellation acknowledgement & a task endpoint accepted or attempted a cancellation-state transition & whether underlying work, descendants, prepared effects, and transferred messages can still commit \\
Process exit & one observed process terminated & whether authority escaped into remote tasks, credentials, triggers, queues, or provider state \\
Token revocation & a named token is no longer accepted under the issuer's revocation semantics & derived credentials, cached authorization, already accepted work, and effects not mediated by that token \\
Exact delegation-edge revocation & the authority consequences of removing a selected graph edge & obligations and carrier effects already emitted before or across the revocation cut \\
Provider-local effect closure & one adapter's registered local frontier is closed at its fence and watermark & compatible cut identity, complete provider coverage, and messages crossing between leaves \\
Closure-evidence verifier & submitted evidence satisfies a declared closure predicate & execution of the root cut, provider actuation, and distributed collection of compatible leaf evidence \\
Distributed termination detection & a computation is passive and its modeled channels contain no messages & a root-specific authority cut, multi-root preservation, and operation-time provider fences \\
Root-scoped quiescence certificate & a complete compatible cutset accounts for cut-relevant old-atom pre-fence acceptances and excludes future protected acceptance under one retired root atom & global business completion, rollback, or effects outside the declared mediation boundary \\
\bottomrule
\end{tabularx}
\end{table}

\subsection{Cancellation is a request, not a closure proof}

Long-running task protocols correctly permit cancellation to be cooperative.
The current MCP Tasks specification treats cancellation as an eventually
consistent task-state operation, and A2A specifies an attempt to cancel rather
than a guarantee that all underlying execution has ceased
\cite{mcp2026tasks,a2a2026spec}. This design is necessary for remote work that
has already entered a provider lifecycle or crossed an irreversible boundary.
It also means that the acknowledgement is an input to quiescence, not the
quiescence certificate itself.

The distinction is phase-sensitive. A proposed task can often be discarded; an
accepted task may already own durable queue state; an executing task may have
delegated credentials; and a completing task may be racing with a protected
commit. Emerging agent-infrastructure guidance likewise distinguishes
revocation semantics across task phases rather than treating a transport close
as task termination \cite{ietf2026agentarchitecture}. The protocol records
these states as authority carriers and obligations and closes them at the sink
that can still accept an effect.

Empirical work on agent-framework stop controls has demonstrated cancellation
orphans, timeout zombies, replayed execution, and sibling leakage, and places
effect enforcement outside the agent runtime \cite{stopmeansstop2026}. That
result establishes the enforcement gap and the value of an external gate. Our
decision object adds a root-selected, multi-provider cut, exact cross-provider
carrier accounting, and explicit preservation of independently authorized
shared work.

\subsection{Credential revocation and bounded expiry}

OAuth token revocation names a concrete credential and gives clients a
standard way to request invalidation \cite{rfc7009}. Short-lived or
heartbeat-renewed credentials can further bound how long authority survives a
lost controller \cite{heartbeat2026}. Usage control treats ongoing
authorization, obligations, conditions, and mutable attributes as first-class
parts of long-lived access, rather than reducing control to the initial
request decision \cite{park2004uconabc}. Time-bounded leases likewise make
continued rights conditional on an explicit term and fault model
\cite{gray1989leases}. These mechanisms are effective leaf
actuators when every use is mediated by the relevant issuer and when derived
credentials and accepted operations are covered.

Root-scoped quiescence quantifies over the larger consequence graph. A token
may have created another token, authenticated a durable job, or authorized a
prepare that no longer consults the token at commit. Revoking the initiating
credential therefore closes one carrier, not necessarily the entire root
cut. The protocol accepts revocation and expiry receipts as local closure
evidence while retaining every uncovered descendant, accepted operation, or
ambiguous completion as open or unknown.

\subsection{Revocation graphs and multiple roots}

Graph revocation asks which principals and rights cease to be authorized when
an edge or grant is removed. Classical trust-management work models policy
changes as explicit state transitions, classifies the semantics of revocation,
and gives formal or rule-based accounts of distributed delegation
\cite{chander2001statetransition,hagstrom2001revocations,
zhang2003delegationrevocation,li2003delegationlogic}. Exact revocation must
preserve authority supported by an alternate valid path; VERA makes this
edge-precise objective explicit \cite{vera2026}. Residual-authorization
analyses show why current reachability alone can be insufficient when the
correct revocation result depends on authorization history
\cite{residualauth2026}.

The protocol adopts the alternate-path requirement and extends the decision
object along the time and effect dimensions. For each live instance it records
an antichain of minimal sufficient root-support sets and a selected witness.
Retiring atom $a_c$ does not delete an instance whose authority can be justified
without $a_c$; instead, continued use requires an atomic rebind to a witness
that excludes $a_c$. The protocol then accounts for obligations and carriers
already emitted under the old selection. Thus graph reachability determines
who may continue, while cut and frontier evidence determines whether old-root
effects can still arrive.

\subsection{Provider effect closure and evidence closure}

Provider-local closure systems reason directly about durable external effects.
The provider-boundary effect-closure model, for example, characterizes
future-use, instance, and lineage closure and exposes a provider-local effect
frontier \cite{santosgrueiro2026effectclosure}.
Such a certificate is a natural provider-frontier leaf: it can attest that a
registered provider adapter installed the named fence and closed every item in
its local coverage manifest.

Local closure is not compositional without a shared cut. Suppose provider
$P$ reports an empty outbound queue immediately before sending a message, and
provider $Q$ reports an empty inbound queue immediately before accepting it.
Both local snapshots can be individually true while the composed system still
contains a live cross-provider obligation. The protocol therefore binds
every leaf to the same root and epoch, records typed SEND and terminal ACK
channel facts, and requires exact token-projected cross-leaf accounting before issuing a global
certificate.

Bounded Agent Closure (BAC) organizes closure evidence for authority,
execution, commitment, and operational state and verifies submitted evidence
against a declared closure specification \cite{boundedagentclosure2026}. That
verification problem is complementary to the execution protocol here.
The protocol establishes the root cut, actuates provider fences, collects
leaf evidence, and constructs a consistent distributed cut; a closure verifier
can validate the resulting evidence object. Keeping these roles separate makes
the trusted boundary explicit: orchestration cannot replace evidence, and
evidence validation cannot by itself perform shutdown.

\subsection{Distributed termination detection}

Classical termination detection asks when all participating processes are
passive and no computation message remains in transit. Dijkstra--Scholten
accounts for diffusing computation, Mattern treats asynchronous and non-FIFO
communication through message-counting algorithms, and consistent-snapshot
techniques make the channel-state problem explicit
\cite{dijkstra1980termination,mattern1987termination,chandy1985snapshot}.
The protocol uses the same core lesson: local emptiness is insufficient when
messages cross a cut.

Authorization quiescence nevertheless differs from computational termination.
A process may remain active under an independent root after the target root is
quiescent; conversely, all processes may be passive while a durable credential,
cron trigger, or prepared provider effect remains able to commit. The relevant
unit is therefore a root-selected authority and effect path, and the barrier is
enforced at protected sinks rather than inferred solely from process states.

\subsection{Three-valued completion}

The protocol reports one of three outcomes. \textsc{Quiescent} means the
complete registered cutset satisfies the certificate predicate.
\textsc{Not-Quiescent} means sound evidence identifies a live authority,
obligation, carrier, or unmatched channel token. \textsc{Indeterminate} means
the result cannot be proved because coverage, fence installation, state
observation, or evidence consistency is missing. This distinction prevents a
timeout, partition, stale cache, or opaque endpoint from being interpreted as
successful shutdown.

\section{System, Threat, and Assurance Model}
\label{sec:system}

\subsection{Assurance statement}

We study shutdown of one authorization root while a long-running agent may
have delegated work across heterogeneous runtimes and providers.  The property
is deliberately root-relative and has two distinct enforcement stages.  A
linearized root-ledger cut retires root-epoch atom \(a\) and immediately
prevents issuer-side creation or expansion of authority under \(a\).  Each
effect provider then installs its own locally ordered barrier.  An already
authorized carrier can reach a provider before that barrier, including after
the root-ledger cut; such an acceptance is admissible only when an exact typed
receipt places it before the provider barrier and the corresponding future-use
frontier is closed.  A certificate establishes that all required barriers are
installed, all such acceptances are accounted for, and no extension can create
a new protected-effect commitment relying on \(a\).  Work supported by a
different current root may continue, and an effect already accepted before its
local barrier may settle.  The certificate therefore means neither that every
process has exited nor that previously committed or physical consequences have
been undone.  It is a positive, mechanically checkable statement about the
absence of residual commit authority from the named root within the declared
enforcement boundary.

The distinction between a request acknowledgement and this assurance is
fundamental.  A runtime may acknowledge cancellation while a queued message,
delegated credential, scheduled trigger, webhook, reservation, or provider
operation remains able to cross its commitment point.  Root-scoped quiescence
requires an authority cut, complete carrier coverage, provider-local effect
closure, and exact accounting for handoffs between covered endpoints.

\subsection{Components and enforcement boundary}

The system contains five logical components.  They may be co-located, but
their evidence roles remain distinct.

\begin{enumerate}
  \item The \emph{root ledger} stores authenticated grants, monotonically
  increasing root epochs, shutdown cuts, and root phases.  It is the
  linearization authority for retiring a root-epoch atom.
  \item The \emph{shutdown coordinator} freezes the coverage manifest for a
  cut, requests endpoint fences, reconciles cross-endpoint channel tokens, and
  composes leaf evidence.  It cannot declare quiescence by fiat.
  \item A \emph{runtime adapter} enumerates the root-dependent carriers under
  its endpoint, preserves their support provenance, and closes, terminalizes,
  or rebinds each carrier through a registered transition.
  \item A \emph{provider/effect adapter} places a monotone fence at the
  authorization-sensitive commitment boundary and emits a provider-local
  frontier certificate.  The frontier orders accepted commitments relative to
  that provider's local barrier; it does not compare heterogeneous provider
  sequence numbers with the root-ledger sequence.
  \item An \emph{independent verifier} checks signatures, epochs, manifest
  coverage, leaf compatibility, support witnesses, fence and log watermarks,
  and channel-token conservation before returning
  \(\mathsf{QUIESCENT}\), \(\mathsf{NOT\_QUIESCENT}\), or
  \(\mathsf{INDETERMINATE}\).
\end{enumerate}

A \emph{carrier} is any durable or live object that can preserve, transmit,
or exercise authority: a task, child process, delegated credential, session,
queue record, timer, cron entry, webhook, reservation, retry record, or
provider operation.  A \emph{channel token} names one cross-endpoint handoff.
The distinction prevents a local scan from treating an emitted but not yet
accounted-for message as absence of work.  A \emph{protected-effect
commitment} is the registered point after which an effect can proceed without
another authorization decision.  Provider profiles locate this point; later
delivery or settlement is not misclassified as a new authorization event.

The enforcement boundary is declared by a versioned manifest.  It lists every
runtime endpoint, provider adapter, effect sink, channel class, and profile
digest reachable by work selected under the named retired root-epoch atom.  Dynamic carriers
need not all exist when shutdown begins: their creation events and incoming
channel tokens must be covered by a manifest endpoint and by the stable-scan
watermark in its leaf certificate.  An endpoint lacking a registered
enumeration, fence, or frontier mechanism remains visible as indeterminate; it
cannot be silently omitted.

\subsection{Adversary and fault model}

The adversary controls the agent's planning and every cooperating descendant
whose behavior is not part of the trusted enforcement boundary.  It may create
deep delegation chains, select one of several authorization roots, move work
between endpoints, schedule delayed callbacks, retain renewable credentials,
and race cancellation against commitment.  It may delay, duplicate, reorder,
or replay messages; crash a coordinator or adapter at any instruction
boundary; exploit stale cancellation acknowledgements; present incomplete or
conflicting inventories; and attempt to relabel root-dependent work as
independently authorized.  Providers and networks may be unavailable or
partitioned for an unbounded interval.

The adversary cannot forge signatures or collision-resistant digests,
rewrite a committed root-ledger record, violate the atomicity of a registered
fence or rebind transaction, or produce a protected commitment outside every
manifested mediation point.  These are explicit assurance premises rather
than properties inferred from agent cooperation.  Compromise of a trusted
commit gateway, provider fabrication of authenticated frontier evidence, and
an effect path absent from the registered manifest are outside the theorem
domain; each is addressed operationally by profile admission and manifest
audit rather than by interpreting missing evidence as success.

\subsection{Assurance assumptions}

Table~\ref{tab:quiescence-assumptions} gives the complete premises used by the
safety results.  Assumptions A1--A9 are safety premises.  A10--A11 are needed
only for the convergence result; safety remains fail-closed during a partition
or permanent provider failure.

\begin{table}[p]
\caption{Assurance assumptions for a registered quiescence profile.}
\label{tab:quiescence-assumptions}
\footnotesize
\begin{tabularx}{\textwidth}{L{0.055\textwidth}L{0.20\textwidth}Y}
\toprule
ID & Assumption & Enforced obligation \\
\midrule
A1 & Issuer mediation and non-expansion & Every transition that issues, delegates, reactivates, or enlarges authority under a root crosses the manifested root issuer and validates the exact root-epoch atom at its root-ledger linearization point.  This premise does not treat exercise by an already authorized carrier as a new issuance. \\
A2 & Authentic canonical evidence & Grants, root atoms, support records, cuts, manifests, fences, channel dispositions, rebind receipts, frontier statements, and certificates use unambiguous canonical encodings and unforgeable authenticators. \\
A3 & Durable root linearization & A cut retiring epoch \(e\), the advance to \(e+1\), the phase transition, the frozen manifest, and the exact cut-target token frontier commit atomically in a serializable durable ledger; epochs never decrease and a retired atom is never reactivated.  The root transaction does not itself install a provider fence. \\
A4 & Support-provenance completeness & Every carrier and handoff records all policy-relevant minimal authorization supports, the witness actually selected for its current authority, and authenticated parent-to-child provenance.  Omission or ambiguity yields an unknown carrier. \\
A5 & Per-sink fence and frontier soundness & Every manifested effect sink \(p\) serializes commitment admission with a durable local barrier \(b_p(c)\).  The barrier causally follows authenticated observation of cut \(c\), rejects every later local attempt selected under the retired atom, and yields a frontier that completely enumerates every cut-relevant acceptance under that atom that was outstanding when \(c\) committed or was admitted from a dependent carrier before \(b_p(c)\), while closing each such commitment's future-use, instance, and lineage frontier. \\
A6 & Registry and manifest completeness & The frozen manifest covers every runtime, provider, effect sink, and channel class reachable from work selected under the retired atom.  Stable scans include all carrier creation and transition events through their stated watermarks. \\
A7 & Exact channel accounting & Each cross-endpoint handoff has a globally unique, retry-stable token whose sole authoritative creation event is \(\mathsf{SEND}(t)\).  For the retired atom, \(\mathsf{SEND}(t)\prec c\) and the token belongs to the exact frontier frozen by \(c\).  Exactly one authenticated terminal disposition---accepted into a covered carrier, rejected by a fence, or expired under an enforced deadline---has the same canonical token projection; acceptance precedes destination-terminal accounting. \\
A8 & Atomic rebind & One atomic transition preserves the exact operation and effect-envelope digests, selects a complete current support witness excluding the retired atom, installs its fresh root-epoch atoms, and emits a single-use receipt binding those facts. \\
A9 & Crash-safe monotonicity & Cut, fence, carrier, channel, and certificate records are durably ordered; an exact retry requires the same request identifier and canonical candidate digest, conflicting identifier/body reuse or certificate reminting is rejected, and recovery never rolls back an epoch, fence, terminal disposition, or closed frontier. \\
A10 & Eventual evidence stabilization & For the liveness theorem, every required sink eventually installs its barrier; every manifested runtime/provider endpoint eventually returns a stable scan and the required frontier or authenticated terminal evidence; every transient pre-certification unknown or state disagreement is eventually resolved.  An authenticated same-key equivocation is never cleared in place: the generation remains indeterminate unless policy starts a new generation under a fresh evidence key. \\
A11 & Finite drain and fair delivery & For the liveness theorem, the cut has finitely many dependent carriers, effects, and tokens; issuer transitions cannot create or expand authority under the retired atom after the cut; every emitted token is eventually delivered to a covered destination or obtains an authenticated fence-reject or enforced-expiry disposition; and every enabled fence, close, typed terminal disposition, or valid rebind action is eventually taken. \\
\bottomrule
\end{tabularx}
\end{table}

\subsection{Assets, attack goals, and verdict discipline}

The protected asset is the ability to create a new protected-effect commitment
using the named retired root-epoch atom.  The adversary succeeds if it causes any of the
following while a verifier returns \(\mathsf{QUIESCENT}\):

\begin{itemize}
  \item an accepted commitment locally ordered after its sink barrier whose
  selected authorization witness contains the retired root-epoch atom;
  \item a cut-relevant acceptance under the retired atom, ordered before its
  local fence, lacking an exact \(\mathsf{PRE\_FENCE\_COMMITTED}\) receipt and
  closed provider frontier;
  \item a root-dependent carrier omitted from every compatible leaf;
  \item a cross-endpoint handoff omitted, duplicated, or accepted into an
  uncovered destination carrier;
  \item a pre-fence commitment used as cover for an additional post-fence
  effect;
  \item a nominal rebind whose replacement witness still depends on the
  retired atom, uses a stale root, or lacks an independent authorization
  decision; or
  \item a certificate composed across different cuts, manifest roots,
  adapter-profile epochs, or conflicting leaf versions.
\end{itemize}

Verdicts separate a demonstrated blocker from an evidence gap.
\(\mathsf{NOT\_QUIESCENT}\) records a verified open carrier, in-flight
token, post-barrier commitment, or other concrete counterexample.
\(\mathsf{INDETERMINATE}\) records missing, stale, opaque, or conflicting
evidence when no verified counterexample is available.  Only
\(\mathsf{QUIESCENT}\) is a release certificate.  Neither of the other
verdicts is converted to success by timeout, retry exhaustion, majority vote,
or an agent's assertion that it has stopped.

The certified statement covers the shutdown generation initiated by the named
cut and the locally ordered sink evidence composed for that generation.  It
does not promise application-level goal completion, global infrastructure
idleness, cancellation of work supported solely by other roots, rollback of a
pre-fence commitment, erasure of information already observed, or reversal of a
physical-world consequence.  These exclusions make the certificate's subject
precise: within the declared profile, the retired atom has no residual path to
a new protected-effect commitment.

\section{Formal Model}
\label{sec:formal}

This section defines the object named by a root-scoped quiescence certificate,
the evidence required to establish it, and nine results used throughout the
paper.  Proofs appear in Appendix~\ref{app:proofs}.

\subsection{Root-epoch atoms and machine state}

Let \(\mathcal R\) be the set of root identifiers and let \(H\) be a
collision-resistant digest over an injective, length-delimited canonical
encoding.  A \emph{root-epoch atom} is
\[
  a=\langle r,e,h_g\rangle\in\mathcal A,
\]
where \(r\in\mathcal R\), \(e\in\mathbb N\) is the root epoch, and \(h_g\)
is the digest of the authenticated grant.  Write \(\operatorname{rid}(a)=r\)
and \(\operatorname{rid}(W)=\{\operatorname{rid}(a):a\in W\}\).  Root atoms are not bearer strings.
An atom is current for a new issuer decision only when the root ledger contains
the same grant digest and epoch in phase \(\mathsf{ACTIVE}\).  A carrier
authorized before retirement can still hold that exact atom while the
distributed generation drains; its subsequent effect admission is governed by
the relevant provider-local barrier, not by reinterpreting the root cut as an
atomic cross-provider fence.

The security-relevant machine state is
\[
 S=\langle A,C,X,T,F,K,L,M,P\rangle .
\]
Here \(A\) is the root ledger; \(C\) the cut ledger; \(X\) the carrier
registry; \(T\) the channel-token ledger; \(F\) the endpoint fence map;
\(K\) the provider commitment/frontier state; \(L\) an authenticated event
log; \(M\) the registered coverage manifests; and \(P\) the versioned adapter
profiles.  Security-changing transitions are serialized within their
respective durable domains and use durable idempotency keys.

Let \(\preceq_r\) be the root ledger's durable order and, for every endpoint or
effect sink \(p\), let \(\preceq_p\) be its durable local admission/log order.
Write \(\prec_r\) and \(\prec_p\) for their strict parts.
The global causal relation \(\prec\) is the transitive closure of local durable
order, authenticated program order, send-before-accept edges, and observation
edges from a committed record to a transition that authenticates that record.
Thus events in different durable domains are compared only through \(\prec\),
never by comparing their implementation-specific sequence numbers.  We use
\(\nu_d(o)\) solely as the local position of event \(o\) within durable domain
\(d\); in particular, \(\nu_r\) and \(\nu_p\) are positions within
\(\preceq_r\) and \(\preceq_p\).

\begin{definition}[Linearized root cut]
\label{def:root-cut}
For a current atom \(a_c=\langle r,e,h_g\rangle\),
\(\mathsf{BeginShutdown}(r,m)\) commits
\[
 c=\langle cid,a_c,e+1,\nu_r(c),h_m,h_P,h_T\rangle
\]
at root-ledger position \(\nu_r(c)\).  The same transaction retires epoch
\(e\) by marking \(a_c\) retired, advances \(A(r)\) to the new epoch \(e+1\)
in phase
\(\mathsf{DRAINING}\), and freezes manifest digest \(h_m=H(m)\), profile
digest \(h_P\), and exact cut-target token-frontier digest
\(h_T=H(T_c)\).  The frontier \(T_c\) is defined below from authoritative
pre-cut \(\mathsf{SEND}\) records.  Advancing the counter does not issue a successor grant; a
later active atom requires a separately authenticated grant transition.
The transaction neither revokes the state of a previously authorized carrier
nor installs a barrier in any provider domain.
\end{definition}

Let \(\mathcal I_r\) contain the issuer events that issue, delegate,
reactivate, or enlarge authority under root \(r\).  A1 requires every
\(i\in\mathcal I_r\) to validate its complete selected witness at its own
root-ledger linearization point.  By contrast, a protected effect has a
registered provider-local \emph{commitment-admission event} \(k\).  For
\(k\) at sink \(p\), \(\nu_p(k)\) is the last mediated local decision after
which the effect can proceed without another authorization decision.  A
settlement, delivery, or physical consequence following \(k\) is part of the
same commitment, not a fresh commitment.  We write
\(a\in\operatorname{auth}(k)\) when the exact selected authorization witness
recorded for \(k\) contains atom \(a\).

\subsection{Minimal support and selected authorization}

A carrier may have conjunctive requirements and alternative roots.  Treating
its roots as a flat union would over-revoke independently authorized work;
treating them as an existential set would permit root laundering.

\begin{definition}[Support antichain]
\label{def:support-antichain}
For carrier \(x\) in state \(S\), its support antichain is a finite family
\[
 \Sigma_S(x)\subseteq 2^{\mathcal A}
\]
such that every \(W\in\Sigma_S(x)\) is a sufficient conjunctive
authorization witness and no distinct \(W,W'\in\Sigma_S(x)\) satisfies
\(W\subset W'\).  The antichain contains all policy-relevant minimal
witnesses admitted by the registered derivation profile.  The durable field
\(\operatorname{sel}_S(x)\in\Sigma_S(x)\) records the witness actually
selected for the carrier's current authority.
\end{definition}

Let \(\operatorname{Current}_S(x,W)\) hold when every atom in \(W\) is current
and the authenticated derivation edges joining those atoms to \(x\) remain
valid.  A selected witness is an authorization fact, not a planner hint.  New
issuer decisions require
\(\operatorname{Current}_S(x,\operatorname{sel}_S(x))\); effect admission by an
already authorized carrier instead presents the exact selected witness to the
sink's barrier rule.  This separation preserves the race interval between the
root cut and provider fence as explicit evidence rather than assuming it away.

For cut \(c\), let \(D_c\) be the least provenance-closed set containing each
carrier present when \(c\) commits whose selected witness contains \(a_c\) and
each carrier later materialized by delivery or acceptance of an outstanding
token \(t\in T_c\) whose support witness contains \(a_c\).  The closure includes
these pre-cut transfers when they are delivered or accepted before the
corresponding channel or ingress barrier, even if acceptance occurs after the
root-ledger cut.  A carrier leaves \(D_c\) only by a
terminal/effect-closed transition or by the atomic rebind of
Definition~\ref{def:rebind}; relabeling a field does not remove the dependency.

For each manifested sink \(p\), let \(K_c^p\) be the cut-relevant commitment
set: every not-yet-closed commitment authorized by \(a_c\) and outstanding at \(p\) when \(c\)
commits, together with every commitment admitted from a member of \(D_c\)
before \(b_p(c)\).  A historical commitment whose future-use, instance, and
lineage frontier was already authenticated as closed before \(c\) is not
reopened merely to populate \(K_c^p\).  This scope captures cut-time provider
state and cut-to-barrier races without requiring an unbounded replay of all
effects ever admitted under epoch \(e\).

\begin{definition}[Atomic independent rebind]
\label{def:rebind}
Let \(op_x\) and \(env_x\) be the carrier's canonical remaining operation and
effect envelope.  They are immutable across a rebind.
\(\mathsf{Rebind}(x,W',c)\) is admissible only if
\[
 W'\in\Sigma_S(x),\qquad a_c\notin W',\qquad
 \operatorname{Current}_S(x,W').
\]
We call such a \(W'\) an \emph{independent support}: it excludes the retired
atom \(a_c\).  A witness using a different root identifier is the common
alternate-root case, not a requirement of the definition.
It re-runs the registered authorization decision for \(x\), atomically
changes \(\operatorname{sel}(x)\) and the carrier authorization epoch while
preserving the exact \(op_x\) and \(env_x\), and emits a single-use receipt
\[
 \rho_x=\mathsf{Sig}
   (cid,xid,H(\operatorname{sel}_{old}),H(W'),H(op_x),H(env_x),
      h_P,\nu_X(\rho_x)).
\]
No descendant inherits the rebind merely from ancestry; it must either be
covered by the same atomic profile transition or carry its own receipt.
\end{definition}

\subsection{Carriers, manifests, and channel conservation}

A carrier record is
\[
 x=\langle xid,kind,endpoint,\Sigma(x),\operatorname{sel}(x),op_x,env_x,
        parents,state,frontier,epoch\rangle ,
\]
where
\[
 \begin{aligned}
 state\in\{&\mathsf{OPEN},\mathsf{EFFECT\_CLOSED},
              \mathsf{PRE\_FENCE\_COMMITTED},\\
            &\mathsf{REBOUND},\mathsf{TERMINATED},\mathsf{UNKNOWN}\}.
 \end{aligned}
\]
For an accepted commitment \(k\) at sink \(p\), a typed pre-fence receipt is
\[
 \begin{aligned}
\rho^{\mathrm{pfc}}_{p,k}=\mathsf{Sig}_p\langle
   &cid,p,kid,H(\operatorname{auth}(k)),H(op_k),H(effect_k),\\
   &\nu_p(k),\nu_p(b_p(c)),s_\rho,
   \mathsf{ACCEPTED\_BEFORE\_FENCE},h_P\rangle .
 \end{aligned}
\]
It verifies jointly with the cut's authenticated fence record only if
\(a_c\in\operatorname{auth}(k)\), the complete selected witness and the
operation and effect digests match the admitted request, the admission and
barrier occur in the same durable domain, \(k\prec_p b_p(c)\prec
\rho^{\mathrm{pfc}}_{p,k}\), both local positions have authenticated log
inclusion, and \(s_\rho\) is the durable receipt-record position.  Thus the
receipt is finalized only after it can bind the installed barrier; it may
account for either a cut-time outstanding acceptance or a cut-to-barrier
acceptance.  State
\(\mathsf{PRE\_FENCE\_COMMITTED}\) means that every
authorization-sensitive commitment associated with the carrier has exactly
one such receipt, is named in the provider frontier, and cannot be extended
into an additional commitment.  It never reclassifies a post-barrier
acceptance.

A manifest
\[
 m=\langle E,Q,G,\Pi,\eta_m\rangle
\]
names runtime/provider endpoints \(E\), cross-endpoint channel classes \(Q\),
commit gateways \(G\), registered adapter profiles \(\Pi\), and the manifest
epoch \(\eta_m\).  Its completeness judgment
\(\operatorname{Complete}(m,a_c)\) states that every commitment path reachable
from \(a_c\) crosses a named gateway and every carrier or handoff on such a
path is enumerable by a named endpoint/channel profile.  This judgment is an
authenticated profile-admission fact checked when the manifest is frozen.
Consequently the named gateways and their barriers hit every registered
directed commitment path reachable from \(a_c\); we call this manifested
hitting set the \(a_c\)-relative path cutset.

Each handoff has a retry-stable token
\[
 t=\langle tid,src,dst,class,H(payload),H(W_t),seq\rangle .
\]
An origin endpoint emits exactly one \(\mathsf{SEND}(t)\).  A destination or
gateway emits exactly one terminal acknowledgement
\[
 \mathsf{ACK}(t,d),\quad
 d\in\{\mathsf{ACCEPT}(x),\mathsf{FENCE\_REJECT},
          \mathsf{ENFORCED\_EXPIRY}\}.
\]
For \(\mathsf{ACCEPT}(x)\), the destination leaf must cover \(x\) and the
accepted carrier's support projection must include \(W_t\).  Fence rejection
and expiry are valid only when the corresponding enforcement evidence is
bound to the cut.  The \(\mathsf{SEND}(t)\) record is the sole authoritative
token-creation event.  For a handoff with \(a_c\in W_t\), A1 and A7 require
\(\mathsf{SEND}(t)\prec c\).  The cut atomically freezes the exact registered
old-atom token multiset
\[
 T_c=\biguplus_{\substack{t:\ a_c\in W_t\\
                 \mathsf{SEND}(t)\prec c}}\{t\},
\]
where each singleton is a one-element multiset and global token uniqueness
makes every multiplicity one.  Its digest is \(h_T\).  Transport may later deliver an already named token,
but it cannot mint or backfill a new old-atom token after the cut.  For
an accepted handoff, the causal order is
\[
 \mathsf{SEND}(t)\prec\mathsf{ACK}(t,\mathsf{ACCEPT}(x))
 \prec \mathsf{TerminalAccount}(x).
\]
Because \(\mathsf{SEND}(t)\) and
\(\mathsf{ACK}(t,d)\) have different types, define the canonical token
projection
\[
 \pi_t(\mathsf{SEND}(t))=t,
 \qquad \pi_t(\mathsf{ACK}(t,d))=t,
\]
and extend \(\pi_t\) elementwise to multisets.  Equality of projected
multisets, together with uniqueness and disposition validity, detects both
loss and duplication without equating heterogeneous records.

\subsection{Endpoint fences and leaf certificates}

For cut \(c\), endpoint fence \(f_p(c)\) contains a provider-local barrier
event \(b_p(c)\) that binds the retired atom, local profile epoch, and durable
fence/log watermark.  The adapter authenticates the root cut before installing
the barrier, giving the causal edge \(c\prec b_p(c)\).  Every commitment
admission at \(p\) is comparable with \(b_p(c)\) in \(\preceq_p\): an accepted
admission authorized by \(a_c\) must satisfy \(k\prec_p b_p(c)\), while every attempt with
\(b_p(c)\prec_p k\) and \(a_c\in\operatorname{auth}(k)\) is rejected.  A
request that straddles cut propagation is classified solely by this local
admission order.  No comparison is made between \(\nu_r(c)\) and
\(\nu_p(k)\).  A stable scan begins after the fence and covers every carrier
and channel transition through its stated watermark; the provider frontier
binds every commitment in \(K_c^p\) visible at that boundary.

\begin{definition}[Leaf certificate]
\label{def:leaf}
A leaf from endpoint \(p\) is the authenticated tuple
\[
 \ell_p=\mathsf{Sig}_p\langle
 \chi,scope_p,\omega_p,\mu_p,f_p,O_p,B_p,R_p,J_p,U_p,v_p,d_p\rangle .
\]
The compatibility header
\(\chi=\langle cid,a_c,h_m,h_P,h_T\rangle\) binds the cut, atom, manifest,
profile family, and frozen token frontier.  \(scope_p\) is the endpoint's manifest slice; \(\omega_p\)
is its evidence-through boundary; \(\mu_p\) its
stable-scan and log-root evidence; \(f_p\) its fence/frontier evidence;
\(O_p\) the multiset of authenticated \(\mathsf{SEND}\) records; \(B_p\) the
multiset of authenticated terminal \(\mathsf{ACK}\) records; \(R_p\) rebind
receipts; \(J_p\) typed pre-fence commitment
receipts; \(U_p\) unknown/conflict facts; \(v_p\in\{Q,N,U\}\) its local
verdict; and \(d_p\) the digest of the canonical evidence snapshot projected
through \(\omega_p\), whose preimage excludes \(d_p\) and the outer signature.
The leaf-record event causally follows
\(\omega_p\), and every referenced barrier, typed receipt, rebind, and
terminal acknowledgement is recorded no later than \(\omega_p\).
\end{definition}

A leaf is locally \(Q\) exactly when its manifest slice and stable scan are
complete, its fence and provider frontier verify, \(U_p=\varnothing\), and
\(J_p\) contains exactly one valid \(\rho^{\mathrm{pfc}}_{p,k}\) for every
\(k\in K_c^p\) preceding \(b_p(c)\) and no receipt for a post-barrier
acceptance.  In addition, every \(x\in D_c\) in its slice is one of:
\begin{enumerate}
  \item \(\mathsf{EFFECT\_CLOSED}\) or \(\mathsf{TERMINATED}\);
  \item \(\mathsf{PRE\_FENCE\_COMMITTED}\) with exact ordered receipts and a
  matching closed frontier; or
  \item \(\mathsf{REBOUND}\) with a valid receipt to a current witness that
  excludes \(a_c\).
\end{enumerate}
An authenticated \(\mathsf{OPEN}\) carrier or accepted post-barrier
commitment is a local \(N\) witness.  A leaf exports its sends and terminal
acknowledgements but
does not decide channel closure in isolation; that decision requires the
composite multiset.  Missing, stale, opaque, or conflicting evidence without a
concrete blocker produces \(U\).

\subsection{Composition and three-valued decision}

Let a partial evidence object be a finite map from manifest leaf identifiers
to leaf digests and bodies, together with a conflict marker.  Define
\(E_1\sqcup E_2\) by keywise union: an absent entry takes the present value;
identical entries collapse; two different authenticated values for one key
produce the conflict marker.  Objects with different compatibility headers
also conflict.  This definition makes retries idempotent while preserving
equivocation evidence.

For compatible evidence \(E\), write
\[
 O(E)=\biguplus_{\ell_p\in E}O_p,
 \qquad B(E)=\biguplus_{\ell_p\in E}B_p.
\]
\(\operatorname{Covered}(E,m)\) holds when the leaf scopes form exactly the
required manifest coverage (overlap is allowed only when the manifest declares
the same replicated evidence key).  \(\operatorname{ChannelsClosed}(E)\)
holds, relative to the cut in the compatibility header, when the projected
sends are exactly the frozen frontier, every send precedes the cut, and every
send has exactly one valid terminal acknowledgement.  Equivalently, token
identifiers and multiplicities reconcile under
\[
  \pi_t(O(E))=T_c=\pi_t(B(E)),
\]
and every accepted token maps to a covered destination carrier.  The predicate
does not assert equality between \(\mathsf{SEND}\) and \(\mathsf{ACK}\)
records themselves.

\begin{definition}[Composite verdict]
\label{def:verdict}
The verifier computes \(V_c(E)\) as follows.
\begin{enumerate}
  \item Return \(N\) if verified evidence contains a concrete open
  root-dependent carrier, an unacknowledged emitted token under a complete
  scan, an accepted commitment \(k\) with
  \(b_p(c)\prec_p k\) and \(a_c\in\operatorname{auth}(k)\), or an accepted
  cut-relevant pre-fence commitment for which a complete authenticated scan
  establishes that its exact ordered receipt or closed frontier is absent.
  \item Otherwise return \(U\) if the header conflicts, the manifest or any
  leaf scope is absent, a signature/profile/watermark is stale, any leaf is
  unknown, an endpoint is opaque, or channel conservation cannot be proved.
  \item Otherwise return \(Q\) if every leaf is locally \(Q\), all fences and
  frontiers verify, coverage is complete, and channels are closed.
\end{enumerate}
\end{definition}

The priority of \(N\) preserves a known counterexample even when other
evidence is unavailable.  Let \(\boldsymbol\omega_E\) be the canonical map
from each manifested leaf identifier to its local evidence-through boundary
\(\omega_p\), ordered by leaf identifier; it reduces to one sequence in a
single-ledger profile.  Define the pre-certificate candidate body
\[
 \mathcal B_q=\langle cid,r,e,h_m,\boldsymbol\omega_E,H(E)\rangle,
 \qquad h_q=H(\mathcal B_q).
\]
The canonical preimage \(\mathcal B_q\) excludes the request identifier
\(qid\), its own digest \(h_q\), issuance metadata, and the outer signature.
Let \(s_c\) be the verifier-ledger position reserved by the atomic issuance
transaction.  Only \(Q\) authorizes issuance of
\[
 \mathcal C_c=\mathsf{Sig}_{V}
 \langle\chi,qid,h_q,H(E),H(O(E)),H(B(E)),s_c,
 \mathsf{QUIESCENT}\rangle .
\]
Certificate issuance is keyed by the canonical request identifier \(qid\) and
the pre-certificate body digest \(h_q\).  After the issuance record
\((cid,qid,h_q,H(E),s_c,\mathcal C_c)\) commits, an exact retry with the same \(qid\) and
\(h_q\) returns the identical \(\mathcal C_c\).  A
different request identifier for the completed cut, or reuse of \(qid\) with a
different body, is a conflicting remint and is rejected rather than signed.

\begin{definition}[Root-relative authorization quiescence]
\label{def:root-quiescence}
A trace is quiescent for cut \(c\), written \(\operatorname{RQ}(c)\), when,
for every manifested sink \(p\), \(c\prec b_p(c)\), every accepted commitment
\(k\in K_c^p\) with \(a_c\in\operatorname{auth}(k)\) satisfies
\[
  k\prec_p b_p(c)
\]
and has an exact \(\rho^{\mathrm{pfc}}_{p,k}\) in a closed provider frontier,
and no accepted commitment whose authorization contains \(a_c\) is ordered
after \(b_p(c)\).  At certificate
issuance, every carrier \(x\in D_c\) is effect-closed, terminal,
\(\mathsf{PRE\_FENCE\_COMMITTED}\), or independently rebound, and all
\(a_c\)-dependent handoffs in the frozen frontier \(T_c\) are terminally accounted for.  Monotone barriers make
the same no-admission-under-\(a_c\) clause hold in every admissible extension of the
certified trace.  A later certificate cannot erase an intervening
post-barrier violation.
\end{definition}

\subsection{Safety, composition, and progress results}

The following nine results characterize the protocol.  The first eight are
safety or algebraic results and do not assume eventual provider availability.

\begin{lemma}[Cut issuer non-expansion]
\label{lem:cut-nonissuance}
Under A1--A3, after \(\mathsf{BeginShutdown}\) retires \(a_c\) at
\(\nu_r(c)\), no later issuer transition in \(\mathcal I_r\) may issue,
delegate, reactivate, or enlarge authority selected under \(a_c\).  The lemma
makes no claim that the root cut itself blocks effect admission by an already
authorized carrier; that obligation belongs to each \(b_p(c)\) under A5.
\end{lemma}

\begin{lemma}[Support-sound projection]
\label{lem:support-projection}
Under A2, A4, A6, and A7, projecting the authenticated provenance graph onto
\(D_c\) includes every carrier and accepted handoff whose currently selected
authority depends on \(a_c\).  Removing nonselected alternative witnesses
does not remove a selected dependency.
\end{lemma}

\begin{lemma}[No root laundering by rebind]
\label{lem:no-laundering}
Under A2, A4, and A8, a carrier leaves \(D_c\) through
\(\mathsf{Rebind}(x,W',c)\) only if its continuing authority has a complete,
current derivation whose selected witness excludes \(a_c\).  Changing labels,
parents, or an unselected alternative cannot establish this condition.
\end{lemma}

\begin{theorem}[Compositional no-false-quiescence]
\label{thm:composition-soundness}
Under A1--A9, if the verifier issues \(\mathcal C_c\), then
\(\operatorname{RQ}(c)\) holds.  In particular, no compatible execution can
contain an unreported cut-relevant pre-fence acceptance under \(a_c\) or an accepted
post-barrier commitment authorized by \(a_c\).
\end{theorem}

\begin{theorem}[Noncompositionality and opacity boundary]
\label{thm:boundary}
The model entails two boundary consequences:
\begin{enumerate}
  \item provider-local \(Q\) leaves do not imply composite \(Q\) without
  exact channel conservation; and
  \item if a reachable endpoint exposes neither complete enumeration nor an
  enforceable fence/frontier, no observer of the remaining evidence can be
  both sound and complete for root-relative quiescence.
\end{enumerate}
The sound verdict for the second case is \(U\).
\end{theorem}

\begin{theorem}[Independent-support preservation]
\label{thm:independent-support-preservation}
Suppose \(W'\in\Sigma_S(x)\), \(a_c\notin W'\), and
\(\operatorname{Current}_S(x,W')\).  A successful atomic rebind to \(W'\)
preserves the carrier's authorization under \(W'\) while retiring \(a_c\),
and the resulting quiescence certificate for \(a_c\) does not revoke any atom
in \(W'\).
\end{theorem}

\begin{theorem}[Merge-order independence]
\label{thm:merge-algebra}
On partial evidence objects sharing a compatibility domain, \(\sqcup\) is
associative, commutative, and idempotent.  Consequently the composite digest
and verdict are independent of leaf arrival order and duplicate delivery;
authenticated disagreement is retained as conflict and therefore cannot yield
\(Q\).
\end{theorem}

\begin{theorem}[Crash/replay stability and unrelated-root locality]
\label{thm:crash-replay}
Under A1--A9, crashes and arbitrary replay of committed protocol
messages cannot decrease the root epoch, remove a fence, reopen a closed
frontier, duplicate a terminal token disposition, reuse a rebind receipt, or
turn a non-\(Q\) evidence object into \(Q\) without new valid evidence.
Moreover, let \(r'\in\mathcal R\), \(r'\ne r\), be absent from the frozen
manifest dependencies and from \(\operatorname{rid}(W)\) for every selected
witness \(W\) represented in \(E\).  Advancing the epoch of, or shutting down,
root \(r'\) does not invalidate \(\mathcal C_c\).
\end{theorem}

\begin{theorem}[Conditional convergence]
\label{thm:conditional-liveness}
Under A1--A11, if every carrier in \(D_c\) eventually becomes effect-closed,
terminal, \(\mathsf{PRE\_FENCE\_COMMITTED}\), or validly rebound; every
transient pre-certification unknown and resolvable state disagreement clears;
and no authenticated same-key equivocation occurs, then the protocol eventually obtains all compatible leaves,
reconciles every token in the frozen frontier \(T_c\), and issues \(\mathcal C_c\).  Until those
premises hold, the verifier remains \(N\) or \(U\) rather than issuing a false
certificate.
\end{theorem}

\section{Root-Scoped Quiescence Protocol}
\label{sec:design}

The protocol turns a shutdown request into a durable, root-relative
security transition. Figure~\ref{fig:root-quiescence-architecture} separates the authority
cut from provider actuation and evidence checking. The cut ledger is the sole
linearization point for retiring an epoch; adapters then install their own
barriers, close local frontiers, and return typed leaves. The aggregator never
infers local facts from a cancellation response, and the checker never invokes
the controller's decision procedure.

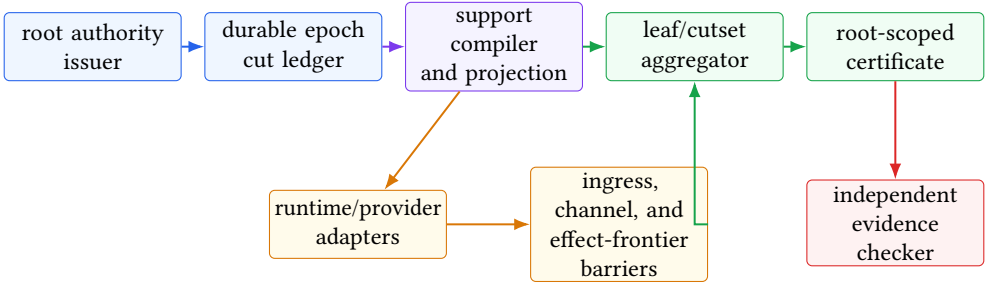
\begin{figure*}[t]
\centering
\begin{tikzpicture}[
  node distance=8mm and 3mm,
  box/.style={draw, rounded corners=2pt, align=center, minimum height=9mm,
    text width=22mm, inner sep=2pt, font=\small},
  flow/.style={-{Latex[length=2.1mm]}, line width=0.75pt},
]
\node[box, draw=protocolBlue, fill=protocolBlueFill] (issuer) {root authority\\issuer};
\node[box, draw=protocolBlue, fill=protocolBlueFill, right=of issuer] (cut) {durable epoch\\cut ledger};
\node[box, draw=protocolPurple, fill=protocolPurpleFill, right=of cut] (support) {support compiler\\and projection};
\node[box, draw=protocolGreen, fill=protocolGreenFill, right=of support] (aggregate) {leaf/cutset\\aggregator};
\node[box, draw=protocolGreen, fill=protocolGreenFill, right=of aggregate] (certificate) {root-scoped\\certificate};

\node[box, draw=protocolAmber, fill=protocolAmberFill, below=13mm of support, xshift=-18mm]
  (adapters) {runtime/provider\\adapters};
\node[box, draw=protocolAmber, fill=protocolAmberFill, right=11mm of adapters]
  (frontiers) {ingress, channel, and\\effect-frontier barriers};
\node[box, draw=protocolRed, fill=protocolRedFill, below=13mm of certificate]
  (checker) {independent evidence\\checker};

\draw[flow, protocolBlue] (issuer) -- (cut);
\draw[flow, protocolPurple] (cut) -- (support);
\draw[flow, protocolGreen] (support) -- (aggregate);
\draw[flow, protocolGreen] (aggregate) -- (certificate);
\draw[flow, protocolAmber] (support) -- (adapters);
\draw[flow, protocolAmber] (adapters) -- (frontiers);
\draw[flow, protocolGreen] (frontiers) -| (aggregate.south);
\draw[flow, protocolRed] (certificate.south) -- (checker.north);
\end{tikzpicture}
\caption{Protocol separation. The root cut retires issuer-side creation and
expansion of old-epoch authority;
adapter-specific barriers and frontier leaves establish local facts; exact
cutset aggregation establishes the composed predicate; a separately
implemented checker validates the evidence object.}
\Description{A two-level flow diagram. The top level moves from the root
authority issuer through a durable epoch cut, support projection, leaf
aggregation, and a root-scoped certificate. Runtime and provider adapters plus
frontier barriers feed evidence upward, and a separate checker validates
the final certificate.}
\label{fig:root-quiescence-architecture}
\end{figure*}

\subsection{Step 1: linearize and seal the root cut}

For root $r$ at epoch $e$, \mcode{BeginShutdown} durably records retirement of
$a_c=(r,e,h_g)$ and advances the issuer to the new epoch $e+1$. The record binds a fresh cut
identifier, the endpoint/sink/channel manifest version, the policy and adapter
profile digests, and the initial graph commitment. The issuer subsequently
rejects any authority-issuing, delegating, reactivating, or enlarging transition
whose selected support contains $a_c$.

This operation is atomic at the authority ledger, not across all providers.
The distributed shutdown generation remains draining until its registered
leaves prove their adapter-specific barriers. An already authorized carrier
can still reach an effect sink before that sink's barrier; the cut neither
erases the carrier nor silently classifies its effect as rejected. Repeating \mcode{BeginShutdown}
returns the same cut; it cannot create another retirement point or decrease the
epoch.

\subsection{Step 2: project actual old-root dependence}

The support compiler evaluates the \emph{selected} witness of each live
authority instance. If it contains $a_c$, every attached carrier and
generator enters the cut's target set. The complete support antichain is used
only to find a possible replacement; its mere existence never edits the
selected witness.

Figure~\ref{fig:support-rebind} shows why the distinction matters. The
antichain $\{\{A\},\{B\}\}$ admits independent continuation, but only after a
compare-and-swap replaces the $A$-selected authorization with a fresh
$B$-selected instance for the same remaining operation and no wider effect
envelope. By contrast, $\{\{A,B\}\}$ is conjunctive: retirement of $A$ removes
the only sufficient witness, so the carrier must drain, terminate, expire under
a trusted bound, or remain nonpositive.

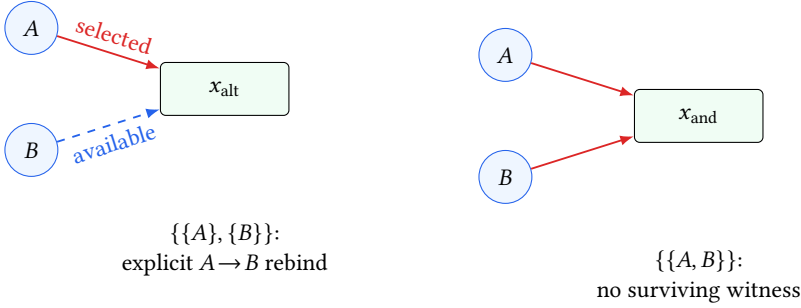
\begin{figure}[t]
\centering
\begin{tikzpicture}[
  every node/.style={font=\small},
  root/.style={circle, draw=protocolBlue, fill=protocolBlueFill, minimum size=7mm},
  work/.style={draw, rounded corners=2pt, fill=protocolGreenFill, minimum width=17mm,
    minimum height=7mm},
  arr/.style={-{Latex[length=1.8mm]}, line width=0.7pt},
]
\node[root] (a1) {$A$};
\node[root, below=9mm of a1] (b1) {$B$};
\node[work, right=17mm of $(a1)!0.5!(b1)$] (alt) {$x_{\mathrm{alt}}$};
\draw[arr, protocolRed] (a1) -- node[above, sloped] {selected} (alt);
\draw[arr, protocolBlue, dashed] (b1) -- node[below, sloped] {available} (alt);

\node[root, right=25mm of alt, yshift=4.5mm] (a2) {$A$};
\node[root, below=9mm of a2] (b2) {$B$};
\node[work, right=17mm of $(a2)!0.5!(b2)$] (and) {$x_{\mathrm{and}}$};
\draw[arr, protocolRed] (a2) -- (and);
\draw[arr, protocolRed] (b2) -- (and);

\node[below=13mm of alt, align=center] {$\{\{A\},\{B\}\}$:\\explicit $A\!\rightarrow\!B$ rebind};
\node[below=13mm of and, align=center] {$\{\{A,B\}\}$:\\no surviving witness};
\end{tikzpicture}
\caption{Minimal-support antichains distinguish alternative from conjunctive
authority. Dashed availability is not authorization until an atomic rebind
selects it.}
\Description{On the left, roots A and B independently support one carrier;
the A edge is selected and the B edge is merely available until rebind. On the
right, A and B jointly support a carrier, so retirement of A leaves no
sufficient witness.}
\label{fig:support-rebind}
\end{figure}

\subsection{Step 3: install adapter-specific barriers}

Each registered adapter maps the logical cut into enforcement at every place
where old-root authority can become durable or cross a protected frontier.
Typical barriers include an issuer tombstone, a task-generation fence, a queue
ingress epoch check, a trigger disable plus generation watermark, a credential
deny epoch, and an effect sink minimum epoch. Every create, accept, trigger,
prepare, and commit path checks the current barrier at its own linearization
point. An acknowledgement emitted before that check is lifecycle evidence, not
closure evidence.

For sink $p$, the adapter first authenticates cut $c$ and then installs durable
local barrier $b_p(c)$, establishing $c\prec b_p(c)$. Commitment admissions and
$b_p(c)$ share the sink's local total order. A cut-relevant acceptance
authorized by $a_c$ and ordered before $b_p(c)$ is first recorded as an
accepted admission. After the barrier is durable, the provider finalizes a
typed receipt binding the cut, sink,
operation and effect digests, complete selected authorization witness,
admission position, barrier position, and receipt-record position. This order
covers both cut-time outstanding acceptances and cut-to-barrier acceptances.
An attempt whose selected support contains the retired atom and
is ordered after $b_p(c)$ is rejected. This local rule
classifies requests concurrent with cut propagation without comparing a root
ledger sequence number to a provider sequence number.

Adapters classify cut-target objects into typed terminal dispositions:
\mcode{RETIRED}, \mcode{FENCED}, trusted \mcode{EXPIRED},
\mcode{PRE\_FENCE\_COMMITTED}, or \mcode{REBOUND}. The pre-fence disposition is
valid only with the ordered receipt and a closed future-use frontier. A prepared or accepted item
is not terminal merely because cancellation was requested. An unavailable
query or unverifiable fence produces \mcode{UNKNOWN}.

\subsection{Step 4: conserve local and transferred carriers}

At any evidence revision, an old-root carrier has exactly one accountable
location: a provider-local registry, a uniquely identified cross-provider
transfer, a terminal/fenced disposition, or a completed rebind. Emission
atomically changes ownership from local to in transit. Acceptance of the same
token changes it from in transit to destination-local; rejection produces a
terminal receipt. The old-atom token-creation record precedes the root cut and
is included in its frozen token frontier; a later delivery cannot create a new
token under that epoch. An accepted acknowledgement precedes any terminal
accounting of the destination carrier. Tokens are exact identities, not an aggregate counter, so
two different messages cannot cancel arithmetically. Since send and terminal
acknowledgement records have different types, composition compares their
canonical token projections, \(\pi_t(O(E))=\pi_t(B(E))\), together with
disposition validity; it never asserts equality of the heterogeneous records.

For FIFO channels, marker order can determine the cut. For non-FIFO channels,
the profile supplies unique sequence identities, a gap-free watermark, and
terminal receipts. A channel that satisfies neither profile remains
\UQ. This rule directly excludes the execution in which both endpoints report
local emptiness while one message lies between their scans.

\subsection{Step 5: produce provider-frontier leaves}

Table~\ref{tab:leaf-schema} lists the minimum logical content of a leaf. A
production encoding can add provider-specific receipts, but it cannot omit a
field used by the composition predicate. The signature authenticates origin
and content; semantic soundness comes from admission of the adapter contract
and its evidence source.

\begin{table}[t]
\caption{Logical leaf-certificate obligations.}
\label{tab:leaf-schema}
\small
\begin{tabularx}{\textwidth}{L{0.20\textwidth} Y Y}
\toprule
Field group & Bound fact & Acceptance condition \\
\midrule
Identity & provider, adapter profile, key identity, schema version & admitted and unambiguous \\
Cut header & root, retired epoch, cut id, cut-record digest & exact agreement across leaves \\
Coverage & manifest/configuration/graph-scope digests; covered ingress, frontier and channels & union equals sealed manifest \\
Barrier evidence & fence receipts, provider-local admission/barrier positions, log watermarks & causal observation of the cut and valid local ordering \\
Local state & open carriers, generators, accepted operations, typed terminal dispositions & no known old-root live blocker; every cut-relevant old-atom pre-fence acceptance has its exact receipt \\
Transfer state & typed SEND and terminal ACK records with unique token maps & exact projected accounting; ACCEPT binds a covered destination carrier \\
Rebind state & exact old/new root-epoch atoms, selected witness, operation and envelope digest & independently current support and atomic replacement \\
Uncertainty & unavailable queries, conflicts, gaps, stale observations & empty for a positive leaf \\
Integrity & evidence root, evidence-through sequence, snapshot digest, issued time, signature & every referenced fence, receipt, and ACK precedes the signed leaf \\
\bottomrule
\end{tabularx}
\end{table}

\subsection{Step 6: compose the cutset}

The aggregator accepts revisions only under an exact-key merge. Repeating an
identical fact is idempotent; two values for the same fact key are conflicting.
Headers, manifests, configurations, cut identities, and profile versions must
agree. Compatible leaf facts form a join-semilattice, so the result is
independent of leaf arrival order, duplication, or aggregation tree.

Known blockers dominate the verdict: any live old-root carrier, unaccounted or
post-barrier accepted operation, generator, or unmatched in-flight token yields
\NQ. In the absence
of a known blocker, missing coverage, a leaf marked unknown, a stale fence, a
token gap, or a configuration conflict yields \UQ. The aggregator returns \Q{}
only when the issuer cut is durable, coverage is exact, every required barrier
is installed, all cut-target carriers have admissible dispositions, channel
tokens close, and every continuation either excludes the retired atom by verified
rebind or is effect-closed.

\subsection{Certificate issuance and replay}

The certificate binds the cut-record digest, graph root, endpoint and sink
catalog digests, fence and terminal receipt roots, alternate-support proof
root, configuration digest, issuance sequence, and certificate identity. Its
issuance transition moves the generation from draining to quiescent in the
same durable transaction. An exact retry carrying the same request identifier
and canonical candidate digest returns the identical committed certificate. A
different request identifier for the completed generation, or reuse of the
identifier with a different body, is rejected as a conflicting remint. A changed manifest, adapter profile,
policy, or root epoch requires a new cut and new evidence.

\section{Protocol Instantiations}
\label{sec:instantiations}

The protocol is provider-neutral, but a leaf is meaningful only through a
concrete adapter contract. This section specifies the executed reference core
and the correspondence obligations that any deployment profile must satisfy.
A logical endpoint shape in the reference artifact is not evidence about a
remote service.

\subsection{Reference ledger and protected sink}

The reference implementation uses a file-backed ledger with canonical JSON,
domain-separated SHA-256 digests, a contiguous hash-linked event log, and
write--fsync--rename--directory-fsync persistence. Mutating transactions take
an exclusive file lock, reload the durable state while holding that lock, and
validate the complete state before committing. Thus, a process that opened the
ledger before a concurrent cut cannot authorize from its stale in-memory
snapshot or overwrite the cut with a lost update. Each root is
\mcode{ACTIVE}, \mcode{DRAINING}, or \mcode{QUIESCENT}. A durable cut increments
the root epoch and freezes the obligation identifiers whose selected witness
contains the retiring root. Sink fences record the minimum admitted epoch; the
reference authorization check requires the exact grant-bound atom set selected
for the obligation, the registered obligation/sink binding, and satisfaction
of the sink-local fence. The root cut prevents creation of new authority from
the retired atom, but deliberately does not stand in for the sink fence: an
already scheduled old-atom attempt remains admissible in the cut-only arm until
the protected sink installs its barrier.

The ledger enforces exact selected witnesses for both alternative and
conjunctive supports, monotone obligation and channel transitions,
conflict-free leaf records, registered global token identity, and current
alternate roots for explicit rebind. Root identifiers in a support template
are instantiated into exact current grant-bound atoms when a witness is
selected or rebound; the alternate-root case advances root $B$ to epoch 2
before rebind and admits only that epoch-2 atom. A terminal channel cannot be
reopened, an unregistered leaf token cannot be introduced as evidence, and a
strict subset of a conjunctive witness cannot authorize an effect.

The reference channel path records a typed \mcode{SEND} binding token identity,
source, destination, payload, authorization atoms, and emission sequence. Its
terminal acknowledgement binds a disposition and evidence; an
\mcode{ACCEPTED} disposition additionally names the destination carrier into
which the token was admitted. A \Q{} leaf is admitted only after the relevant
sink fences, obligation receipts, and terminal channel acknowledgements have
already entered the durable log. The leaf binds an
\mcode{evidenceThroughSequence} and a digest of the endpoint-local evidence
snapshot through that sequence, so a later record cannot retroactively satisfy
an earlier positive leaf.

A pre-fence commitment is terminally accounted only by a typed receipt created
after the sink barrier and binding the cut, obligation, operation, effect,
sink, exact selected authorization atoms, admission sequence, and barrier
sequence. The executed pre-fence case covers both an accepted admission that
is already outstanding when the cut occurs and a second admission accepted
between the cut and barrier; each receives its own exact post-barrier receipt.
Certificate retry identity is the pair
$(\mathit{requestId},\mathit{candidateDigest})$: an exact pair returns the same
committed certificate, reuse of the request identifier with another candidate
body is rejected, and a different request identifier cannot remint the same
generation. This is a provider-free executable semantics: no network request,
cloud resource, credential, payment, or production effect is performed. A
protected effect is an inert typed record whose acceptance is decided at the
local sink.

\subsection{Cancellation-race worker}

The late-effect test suite forks an independent worker process. In every causal-separation
case, the worker first acknowledges that the protected unit has been scheduled;
only then may the controller cut or fence. The harness next observes an explicit
\mcode{EFFECT\_READY} barrier, sends \mcode{CANCEL}, and requires the first
matching post-cancel arrival to be \mcode{CANCEL\_ACK}; an earlier effect arrival
fails the case. Each case also emits a case-bound, nonce-bearing, hash-linked
harness transcript. The worker retains the already scheduled unit and, only
after enqueuing the acknowledgement, presents the same old-atom effect to the
parent-side protected sink. The three arms differ only in the intervention applied after scheduling:
(i) cancellation only leaves the old effect admissible; (ii) a root cut without
a sink fence prevents new authority expansion but still leaves that already
scheduled effect admissible; and (iii) the same cut followed by the sink-local
fence rejects the identical effect at its linearization point. The design
therefore identifies the fence, rather than cancellation timing or work
creation timing, as the cause of rejection.

Durability cases reconstruct the ledger after fencing and exercise a process
that opened the ledger before the cut. The latter must reload the cut under the
transaction lock before evaluating its attempt. Further cases exercise dual
pre-barrier admission accounting with post-barrier receipts, leaf snapshot
ordering, typed channel \mcode{SEND}/terminal-acknowledgement accounting, exact
certificate retry versus body conflict and remint, epoch-2 alternate-root
rebind, conjunctive-witness enforcement, ghost leaf-token rejection, and the
irreversibility of terminal channel state. The checker validates the recorded
transcript chain, nonce, registered message order, rejection entries, and
their bindings to durable events. The transcript is evidence about this
harness execution; it is not presented as an observation of an external
provider. These executions test concrete process, ordering, and persistence
behavior while remaining isolated from external providers.

\subsection{Endpoint correspondence contracts}

Table~\ref{tab:endpoint-contracts} defines representative correspondence
contracts for concrete adapters. MCP and A2A cancellation responses retain
their specified lifecycle meanings; an adapter may add only evidence obtained
from authenticated query, fence, receipt, and sink interfaces
\cite{mcp2026tasks,a2a2026spec}. A profile that cannot supply a required field
must report it as unavailable, causing the affected leaf to remain \UQ.

\begin{table}[t]
\caption{Endpoint correspondence contracts for conforming deployments.}
\label{tab:endpoint-contracts}
\small
\begin{tabularx}{\textwidth}{L{0.17\textwidth} Y Y Y}
\toprule
Endpoint & Persistent carrier & Required barrier/evidence & Nonpositive condition \\
\midrule
MCP asynchronous task & task state, provider work, descendants & cut-bound cancel receipt, terminal query, descendant/transfer coverage, protected sink fence & cooperative ack without terminal/frontier evidence \\
A2A delegated task & remote task, artifact, delegated child & task correlation, terminal or fence receipt, outbound child tokens & unsupported cancel, missing child catalog, late artifact path \\
Credential issuer & access/refresh token and derived credentials & issuer tombstone/introspection, propagation bound, derived-token closure & unknown derivation, cache bypass, untrusted expiry \\
Durable queue or trigger & message, retry, cron, webhook, generator & ingress generation fence, gap-free watermark, accepted-message terminal receipts & forgotten generator, watermark gap, unfenced consumer \\
Effect reservation/gateway & reservation, prepare, settlement & cutoff-aware reserve/commit gate and terminal/no-effect receipt & ambiguous prepare/commit, bypassable effect sink \\
\bottomrule
\end{tabularx}
\end{table}

\subsection{Runtime shutdown IR}

A runtime adapter maps native task, session, and tool events into the following
shutdown IR:

\[
\begin{aligned}
\langle{}&\mathit{targetRoot},\mathit{rootEpoch},\mathit{caller},
\mathit{taskCorrelation},\\
&\mathit{endpointCatalogDigest},\mathit{sinkCatalogDigest},
\mathit{barrierPolicyDigest},\mathit{runtimeBoundary}\rangle .
\end{aligned}
\]

The IR does not treat a runtime as providing root-scoped quiescence by itself.
It makes correlations and omitted evidence mechanically visible. A conforming
mapping freezes its source boundary, produces the complete normalized object,
and fails closed when a root, epoch, task, endpoint, sink, policy, or runtime
identity field is absent or misbound. Runtime-specific lifecycle events remain
inputs to the provider and sink contracts; they do not replace them.

\subsection{External-field correspondence}

A provider profile classifies each protocol or service field as direct,
derived under an explicit rule, supplementary deployment evidence, or
unavailable. MCP task state, A2A lifecycle state, OAuth revocation and
introspection, and queue or scheduler receipts each cover only part of the leaf
schema \cite{mcp2026tasks,a2a2026spec,rfc7009,rfc7662}. A profile cannot treat
the union of unrelated documents as one implementation: it identifies the
concrete source for every required field, and any unavailable field keeps the
corresponding leaf \UQ.

\subsection{Cryptographic and trust profiles}

The reference artifact uses content digests to make deterministic replay and
tamper checking explicit; it does not present those digests as provider
authentication. A deployment profile additionally binds leaf signing keys,
key status, signature algorithm, replay domain, certificate schema, and
provider identity. Even a valid signature proves only which admitted principal
made the assertion. Coverage and semantic truth continue to depend on complete
mediation, authoritative observations, and the leaf contract stated in
Section~\ref{sec:system}.

\section{Evaluation}
\label{sec:evaluation}

The evaluation tests the provider-free durable-ledger artifact, including
strict causal separation among cancellation-only, cut-only, and
cut-plus-fence control. It evaluates the executable protocol invariants and
does not transfer those results to unexecuted runtime or provider interfaces by
renaming a logical endpoint as a deployed service.

\subsection{Research questions}

\begin{description}[leftmargin=3.1em,style=nextline]
  \item[RQ1: Separation.] Can an acknowledged cancellation still yield a
  protected effect under the baseline, while a root cut and sink fence reject
  the identical old-epoch attempt?

  \item[RQ2: Decision soundness.] Do known live obligations and channels remain
  \NQ, and do missing endpoint or fence facts remain \UQ?

  \item[RQ3: Precision.] Can shutdown of root $A$ preserve the same remaining
  work under an independently sufficient current root $B$ through an explicit
  exact-witness rebind, without reusing $A$'s grant-bound authorization atom?

  \item[RQ4: Durability and evidence integrity.] Do persisted fences survive
  restart and stale-process attempts; do transaction locks prevent lost
  updates; and do post-barrier exact commitment receipts, evidence-bounded
  leaves, typed channel sends and terminal acknowledgements, exact
  request/body retry semantics, remint rejection, channel monotonicity, and a
  separately implemented checker preserve the certificate boundary?

\end{description}

\subsection{Executable 17-case late-effect test suite}

The refreshed corpus contains 17 cases (Table~\ref{tab:late-effect-test-corpus}). Each
case runs in a fresh temporary ledger. For the three causal-isolation arms, the
worker reaches a \mcode{STARTED} barrier before cancellation, cut, or fence, so
the protected unit is scheduled before the interventions diverge. The harness
then observes \mcode{EFFECT\_READY}, sends cancellation, and accepts only an
actual \mcode{CANCEL\_ACK}-before-effect arrival order. Every case
also carries a case-bound nonce and hash-linked harness transcript. The runner
embeds the registration for reporting, while the checker holds a separate
17-case oracle and imports no controller, corpus, worker, or adapter module. It
independently recomputes the transcript and ledger hash chains, leaf evidence
snapshot digests, certificate semantic snapshot digests, and result-set digest.
It then checks the registered transcript order and rejection entries, their
bindings to durable effect/intervention records, the root generation, exact
selected authorization atoms, and each effect decision reconstructed from its
obligation binding and event-time sink fence. It further checks cut-bound fence records, post-barrier commitment
receipts, leaf evidence boundaries, typed channel sends and acknowledgements,
post-cut effects, send-before-cut and accepted-ack-before-destination-terminal
causality, channel monotonicity, and certificate request/body identity.
This reconstruction establishes consistency and registered semantics of the
serialized evidence; process scheduling before it was recorded and behavior
of unexecuted external services are outside that checker's observation domain.

\begin{table}[t]
\caption{Provider-free 17-case late-effect test-suite registration.}
\label{tab:late-effect-test-corpus}
\footnotesize
\begin{tabularx}{\textwidth}{L{0.31\textwidth} L{0.20\textwidth} L{0.18\textwidth} Y}
\toprule
Scenario group & Cases & Registered outcome & Registered distinguishing fact \\
\midrule
Cancellation acknowledgement only & 2 & baseline false-stop & late epoch-1 effect accepted after acknowledgement \\
Root cut without sink fence & 1 & \NQ & already scheduled old-atom effect remains admissible \\
Root cut and sink fence & 2 & \Q & identical already scheduled effect rejected at the sink \\
Restart and stale-process interleaving & 2 & \Q & persisted barrier is reloaded under the transaction lock \\
Alternate-root rebind & 1 & \Q & exact current root-$B$ epoch-2 atom replaces, rather than edits, root-$A$ authorization \\
Opaque endpoint / missing fence & 2 & \UQ & absent evidence never becomes a positive certificate \\
Typed pre-fence commitments & 1 & \Q & two exact receipts, finalized after the barrier, account for pre-cut-outstanding and cut-to-barrier admissions \\
Open obligation / in-flight token & 2 & \NQ & known carrier blocks positive certification \\
Exact retry / body conflict / remint & 1 & replay safe & the same request/body pair is idempotent; same-ID/different-body and different-ID attempts are rejected \\
Conjunctive witness enforcement & 1 & rejected & a strict subset of the selected conjunction cannot authorize \\
Ghost leaf token & 1 & rejected & an unregistered token cannot appear in frontier evidence \\
Terminal channel reopen & 1 & rejected & distinct sender/receiver leaves reconcile a pre-cut send and accepted terminal ACK bound before destination termination; terminal state is irreversible \\
\bottomrule
\end{tabularx}
\end{table}

The worker is a real child process, but endpoint labels in this layer denote
logical shapes, not claims about remote MCP, A2A, OAuth, or scheduler
implementations. Effects are inert typed records evaluated by a local protected
sink. The stale-process case opens a second ledger instance before the cut and
attempts the protected effect only after the first process commits the cut and
fence. The case fails if the second process can authorize from its stale
snapshot or overwrite the intervening state. The pre-fence case separately
admits one effect before the cut and one after the cut but before the barrier,
then requires both exact receipts to be durably finalized after the barrier
before its positive leaf. The channel case records a typed send between distinct
logical endpoints and an \mcode{ACCEPTED} terminal acknowledgement naming the
destination obligation before that obligation's terminal receipt and the two
endpoint leaves. This construction isolates causal, accounting, and
persistence distinctions while keeping the experiment deterministic and
provider-free.

\subsection{Comparison boundary}

The executed baselines are cancellation acknowledgement alone and a root cut
without a sink fence. They share the same scheduled unit and protected-effect
oracle with the complete cut-plus-fence arm, so the comparison isolates the
mechanism responsible for rejecting the late effect. Process exit, token
revocation, subtree cascade, graph-targeted revocation, and conjunctions of
endpoint-terminal states are analyzed through their stated semantics and the
formal counterexamples in Sections~\ref{sec:background} and
\ref{sec:formal}; they are not reported as measured remote-service baselines.
A provider-local closure verifier remains a compatible leaf mechanism rather
than a substitute for the root-scoped composition protocol.

\subsection{Adversarial checker evaluation}

The separate checker additionally faces 44 semantically modified and
consistently rehashed traces. The mutations cover accepted-effect and
fence safety; effect-code, operation, carrier, and authorization-atom binding;
cut-frontier completeness and root-epoch causality; exact pre-fence accounting;
leaf, certificate, and request identity; typed channel payload, destination,
and terminal-state integrity; and complete worker, retry, and rejection
transcripts. Cross-case substitutions are resealed under the receiving case,
so rejection depends on semantic reconstruction rather than a stale digest.
The suite also rejects a SEND moved after its frozen cut and an accepted
acknowledgement moved after destination-terminal accounting.

\subsection{Metrics and acceptance criteria}

Primary registered outcomes are false-quiescence certificates,
post-certificate old-root frontier crossings, unknown-to-Q collapses, residual
authority and obligation counts, exact phase/code matches, alternate-root
preservation, and replay/tamper detection.

The artifact gate requires 17/17 registered runner outcomes, 17/17 separate
checker decisions, rejection of all 44/44 semantically rehashed
regressions, zero false certificates, and zero registered post-certificate
old-root crossings. Every insufficient-evidence case must remain \UQ; every
known live carrier must remain \NQ; and every registered independently
authorized alternate-root case must be preserved. Performance is descriptive
and does not serve as evidence for safety.

\section{Executable Evidence}
\label{sec:evidence}

The evaluated system is the executable durable-ledger artifact, its 17-case
late-effect test suite, and the separately implemented checker. Numerical outcome claims
enter this section only from fresh runner and checker receipts whose digests
bind the complete trace set. The refreshed runner matches all 17 registered
outcomes, and the refreshed separate checker agrees on all 17 traces and
rejects all 44 semantically modified and consistently rehashed regressions.
Table~\ref{tab:initial-results} reports this executable evidence.

\begin{table}[t]
\caption{Fresh-run evidence for the executable 17-case layer.}
\label{tab:initial-results}
\small
\begin{tabularx}{\textwidth}{Y L{0.23\textwidth} L{0.28\textwidth}}
\toprule
Registered check & Fresh result & Status \\
\midrule
Expected case verdict and effect outcome & 17/17 & pass \\
Separate checker agreement & 17/17 & pass \\
Cancellation-only / cut-only late effects accepted & 2/2 and 1/1 & pass \\
Cut-plus-fence, restart, and stale-process late effects rejected & 2/2, 1/1, and 1/1 & pass \\
Root decision totals & $\Q$: 6; $\NQ$: 3; $\UQ$: 2 & pass \\
Alternate-root epoch-2 rebind / dual pre-fence accounting & 1/1 and 2/2 admissions receipted & pass \\
Exact retry / same-ID body conflict / different-ID remint & 1/1, 1/1, and 1/1 & pass \\
Conjunctive witness / ghost token / terminal reopen & 1/1 each & pass \\
Evidence-bounded positive leaves & 11/11 & pass \\
Typed channel send / terminal ACK destination binding & 1/1 and 1/1 & pass \\
Semantically rehashed checker regressions rejected & 44/44 & pass \\
\bottomrule
\end{tabularx}
\end{table}

\subsection{Causal-isolation criterion}

All three arms use the same ordering: the worker reports \mcode{STARTED} before
intervention, reports \mcode{EFFECT\_READY} before cancellation, and enqueues
\mcode{CANCEL\_ACK} before emitting the same old-atom effect. The parent consumes
the first matching post-cancel arrival without type-selective reordering. Both
cancellation-only cases acknowledged intent and accepted the effect. The
cut-only case prevented new authority expansion but accepted the already
scheduled effect because no sink barrier existed. Both cut-plus-fence cases
rejected that effect at the protected sink. The restart and stale-process cases
also rejected it, for four mechanism rejections in total. The comparison
therefore attributes rejection to the sink-local fence, rather than to
cancellation timing, process exit, or post-cut work creation.

\subsection{Typed decision and identity criteria}

The two insufficient-evidence cases remained \UQ. The open obligation,
in-flight channel, and accepted cut-only effect remained \NQ. Six cases reached
\Q, and no nonpositive case carried a root-$A$ quiescence certificate.

The alternate-root case begins with support antichain
$\{\{A\},\{B\}\}$ and an exact selected root-$A$ witness. The execution first
advanced active root $B$ to epoch 2. The rebind then atomically instantiated
the alternative support as the exact current root-$B$ epoch-2 atom before
root-$A$ certification and permitted the protected effect only under that atom.
The conjunctive case rejected an obligation selected under $\{A,B\}$ when it
presented only $B$. These cases distinguish independent preservation from
authority laundering by label deletion or stale-epoch substitution.

The pre-fence case admitted two effects under the retiring atom: one before the
cut whose commitment remained outstanding at the cut, and one after the cut
but before the sink barrier. Only after the barrier did the ledger finalize a
separate typed receipt for each admission. Each receipt binds the cut,
obligation, operation, effect, sink, exact selected authorization atoms,
admission sequence, and barrier sequence; the case reached \Q{} only after both
receipts were present.

Each positive leaf binds the last preceding evidence sequence and a digest of
its endpoint-local evidence snapshot. The ledger admits such a leaf only after
the applicable fence, obligation receipt, and typed terminal channel
acknowledgement. The terminal-channel case recorded a typed send from one
logical endpoint to another before the cut, followed by an \mcode{ACCEPTED}
acknowledgement bound to the destination obligation before that obligation's
terminal receipt and before both endpoint leaves; the later attempt to reopen
the token was rejected. Leaf admission
also rejected the unregistered ghost token.

\subsection{Durability and checker separation}

The restart case reconstructed the ledger after the cut and fence and rejected
the old-atom attempt. The stale-process case opened a second instance before
the cut; its later transaction reloaded the durable state under the exclusive
lock and rejected the same class of attempt. Together they exercise
persistence, stale-snapshot exclusion, and lost-update prevention at the
authorization boundary.

The certificate identity case returned the same committed certificate for an
exact retry of the same $(\mathit{requestId},\mathit{candidateDigest})$ pair.
It rejected reuse of that request identifier with a different candidate body
as a request conflict and rejected a different request identifier as a remint
of the committed generation. The ledger's state validator accepted the
contiguous event sequences and predecessor digests in every final state.

The checker carries its own case oracle and implements canonicalization and
domain-separated hashing independently. It neither imports nor calls the
ledger, runner, worker, or corpus. It recomputes every durable-event and
harness-transcript hash-chain link, every leaf evidence-snapshot digest, every
certificate semantic digest, every case digest, and the result-set digest. Its
case rules reconstruct each effect decision from its obligation binding and
event-time sink fence, as well as the cut frontier, exact current support atoms, receipt
and fence ordering, typed channel send/acknowledgement bindings, destination
carrier for an accepted acknowledgement, send-before-cut and
accepted-ack-before-destination-terminal causality, certificate request/body identity,
and registered negative-operation outcomes. Its refreshed oracle agreed on
17/17 traces and reproduced the runner's result-set digest. Agreement is
therefore a distinct implementation check, not a second invocation of the
controller's verdict function.

For child-process and rejected-operation evidence, the checker validates the
case-bound nonce, transcript chain, registered entry order, rejection code,
and correspondence to the durable intervention/effect records. It does not
claim an independent observation of facts outside the serialized execution,
such as behavior of a remote provider that this provider-free harness never
invoked.

A separate adversarial regression harness modifies certified traces and then
recomputes the affected event chain, leaf and certificate digests, transcript
chain where applicable, case digest, and result-set digest. The checker
rejected 44/44 such traces. The registered mutations exercise effect/fence
decision reconstruction; cut, root-epoch, obligation, receipt, leaf, and
certificate identity; cut-frontier and negative-leaf completeness; exact
alternate and conjunctive support; pre-cut and cut-to-barrier accounting;
channel authorization, payload, destination, SEND, acknowledgement, and
terminal causality; and worker, stale-process, retry, and rejected-operation
transcripts. Several mutations splice internally well-formed fields from a
different registered case and then reseal every affected digest. The checker
therefore rejects the resulting trace by independently reconstructed semantics,
not because the mutation left a stale outer hash. In particular, it rejects a
SEND moved after the frozen cut and an accepted acknowledgement moved after the
destination carrier's terminal receipt. This gate prevents a later certificate
from concealing an earlier registered safety violation even when the modified
trace is rehashed consistently.

\subsection{Claim discipline}

The executable evidence comprises the durable ledger, protected-sink
authorization path, child-process race harness, 17-case registration, and
separate checker described above. Its refreshed runner matched 17/17 registered
outcomes, the checker agreed on 17/17 traces, and the checker rejected 44/44
semantically rehashed regressions. The executable evidence further includes exact post-barrier
accounting of both registered pre-barrier admissions, evidence-bounded positive
leaves, typed channel send/terminal-acknowledgement binding, exact-current
epoch-2 alternate-root rebind, request/body-bound certificate retry, and
case-bound hash-linked harness transcripts.

The cross-provider composition results follow from the formal model and proof
obligations in Section~\ref{sec:formal}. The provider-free artifact supplies an
executable instantiation of the ledger, cut, fence, carrier, channel, rebind,
leaf, and certificate invariants. Endpoint names in its 17-case corpus are
semantic categories, not assertions of runtime or remote-provider conformance.
A deployment claim requires concrete adapters to satisfy the correspondence
contracts in Section~\ref{sec:instantiations}; no such claim is inferred from
this artifact. This scope preserves a one-to-one mapping between every
empirical sentence and the system that produced its receipt.

\section{Related Work}
\label{sec:related}

\subsection{Task cancellation and credential revocation}

MCP Tasks and A2A expose asynchronous task lifecycles and cancellation
operations \cite{mcp2026tasks,a2a2026spec}. Their cooperative semantics avoid
promising an impossible instantaneous rollback once remote work has advanced.
The protocol treats their acknowledgements and terminal states as evidence
about registered task carriers, then additionally closes derived authority,
provider obligations, and cross-provider messages at a root-epoch cut.

Stop Means Stop measures cancellation and timeout enforcement gaps in agent
frameworks and evaluates an external effect gate with fence-on-cancel
\cite{stopmeansstop2026}. The protocol takes external mediation as an
admitted leaf capability and targets the distinct composition question: one
retiring authorization root across heterogeneous providers, in-flight
channels, and shared carriers with independently sufficient support.

OAuth token revocation standardizes invalidation of a named access or refresh
token \cite{rfc7009}. Heartbeat and short-lived credentials provide a temporal
bound on authority persistence when renewal stops \cite{heartbeat2026}. These
are useful closure mechanisms for credential leaves. The compositional
quiescence predicate is broader: it covers effects already accepted under the
credential, descendants and triggers derived from it, and continuing shared
work whose support must be rebound to another root.

\subsection{Agent tool execution and adversarial inputs}

ToolEmu demonstrates scalable risk testing for language-model agents operating
over consequential tool surfaces \cite{ruan2024toolrisks}, while AgentDojo
evaluates tool-using agents over untrusted data under prompt-injection attacks
\cite{debenedetti2024agentdojo}. These systems establish the practical
importance of adversarial tool execution but do not define an
authorization-retirement predicate. The protocol therefore treats plans and
tool-supplied inputs as untrusted and places its claim at authenticated,
mediated effect sinks.

\subsection{Temporary authority and the effect boundary}

Lingering Authority makes temporary resource/effect capabilities explicit in
a request--grant--invoke monitor and rejects replay of epoch-bound handles once
their grant episode closes \cite{lingeringauthority2026}. Its guarantee assumes
a sound typed catalog and complete mediation. It also states the precise
concurrency boundary relevant here: a call permitted before closure but still
in flight must be serialized before the effect or revalidated immediately at
that effect. The protocol supplies the subsequent root-wide accounting
layer. A retiring grant atom is cut across a declared set of heterogeneous
endpoints; work may remain live only through an independently sufficient
alternate-root witness; each cut-relevant carrier must be closed or rebound;
and each cut-relevant old-atom acceptance preceding its local barrier must
receive an exact typed disposition in a provider leaf.

Temporary Authority, Permanent Effects defines commit-time authorization in
terms of witness freshness, causal priority, effect binding, and eligibility,
and evaluates a fail-closed monitor on protected commit surfaces
\cite{temporaryauthority2026}. That rule is a natural sink-adapter obligation
for the root-scoped quiescence protocol. It decides whether one proposed durable effect may cross
its commit boundary; the root-scoped protocol additionally establishes which
authorization epoch is retiring and composes the outstanding carriers,
provider frontiers, and cross-provider messages needed to certify that no
further commit attributable to that root remains possible.

AID-Guard carries authorization through one provider-atomic effect and one
reservation lineage, revalidates the immutable request and provider state at
commit, and retains ambiguity until a terminal result or certified no-effect
with a delivery fence permits release or one successor
\cite{aidguard2026}. Its analysis explicitly treats revocation ordered after a
durable submission as non-retroactive: the submitted request must be
reconciled, not presumed cancelled. This is exactly the class of obligation a
provider-frontier leaf must settle. The composition unit differs:
AID-Guard's conditional uniqueness property is per reservation and supported
provider contract, whereas the protocol retires one selected root across a
coverage manifest, atomically rebinds shared work to an independently
sufficient root, and reconciles tokens crossing provider leaves before issuing
the root-relative certificate.

AIRGuard derives step-level authority and enforces normalized tool actions at
runtime. EBL-Core binds one fully materialized candidate to typed evidence and
a verifiable Decision Derivation in an Execution Release Contract, then
governs a separate single-use Execution Grant whose redemption and revocation
are linearized \cite{airguard2026,eblcore2026}. These systems reinforce the
operation-time mediation expected of a registered protected sink. Their decision unit is an action
or grant; root-scoped quiescence additionally freezes one authorization epoch,
accounts for its already admitted commitments and durable carriers, and
composes closure evidence across the declared provider cut.

\subsection{Revocation and residual authority}

Classical authorization research established the policy layer on which agent
revocation builds. Delegation Logic casts distributed authorization as a
proof-of-compliance problem over policies and credentials
\cite{li2003delegationlogic}. State-transition trust management makes
policy-changing actions explicit \cite{chander2001statetransition}; revocation
classifications separate distinct semantic choices
\cite{hagstrom2001revocations}; and RDM2000 gives a rule-based account of
multistep role delegation and revocation \cite{zhang2003delegationrevocation}.
DW-RBAC specializes delegation and revocation to workflow tasks
\cite{wainer2007dwrbac}.
These models determine which authority survives a policy transition. The
protocol adds a different post-decision obligation: it must drain, fence, or
rebind carriers and effects already admitted under the retired root across a
declared provider cut.

VERA targets exact graph revocation and preserves authorization that remains
reachable through valid alternate roots \cite{vera2026}. ResidualAuth studies
the state that must be retained to compute correct authorization after
revocation \cite{residualauth2026}. The protocol shares the precision goal
but asks a subsequent distributed question. After choosing what remains
authorized, can any carrier or in-flight obligation selected under the retired
root still produce a protected commit? The support antichain and atomic rebind
receipt connect exact multi-root preservation to the epoch cut; provider
frontiers and channel tokens then cover consequences released before their
provider-local barriers, including consequences admitted while the root cut
was propagating.

Bounded Agents' Agentic Principal Chain carries restricted authority and
cross-step composition state along delegation, and Earned Authority constrains
evolving agents with a fixed effect ceiling and evidence-governed changes
below it \cite{boundedagents2026,earnedauthority2026}. Both make authorization
history and non-amplification explicit while a grant remains usable.
The protocol addresses the complementary retirement predicate: after a
particular root atom is cut, it certifies that no manifested carrier or
in-flight transfer can use that atom for a future protected acceptance.

This difference separates the protocol from subtree cascade. Cascading along
all descendants can eliminate some late effects, but it over-revokes shared
principals and still need not close a message, reservation, or prepared commit
already transferred outside the enumerated subtree. The protocol closes the
root-selected effect consequence while preserving a principal only through a
new selected witness that excludes the retired atom.

\subsection{Agent and provider closure}

Distributed proof construction shows how authorization evidence can be
assembled across credential holders without collapsing policy proof into a
central oracle \cite{bauer2005distributedproving}. The verifier here also
composes authenticated evidence, but its evidence concerns post-cut carrier
and effect closure rather than proof that a new access request is authorized.

Bounded Agent Closure provides a closure vocabulary spanning authority,
execution, commitment, and operational state, together with verification of
submitted closure evidence
\cite{boundedagentclosure2026,risu2026boundedclosurespec}. The protocol is
an active execution protocol for producing a particular root-scoped evidence
object across heterogeneous endpoints. It begins with an atomic epoch cut,
performs issuance freeze and sink fencing, requires explicit alternate-root
rebind, and composes compatible provider leaves with exact channel accounting.
The two abstractions consequently meet at an evidence interface: closure
criteria define what a leaf must prove, while the protocol establishes and
collects those facts under crash, replay, and partition.

The provider-boundary effect-closure work formalizes provider-local effect
closure, including future-use, instance, and lineage closure and an explicit effect frontier
\cite{santosgrueiro2026effectclosure}. The protocol uses this class of result as a leaf
certificate rather than reproducing provider semantics at the global layer.
Its contribution is the cross-provider root cutset: manifest completeness,
root/epoch compatibility, fence watermarks, alternate-root receipts, and
balanced inter-leaf channel tokens. The construction also shows why the union
of locally quiescent statements is unsound without a consistent cut.

CONTINUITY specifies assume--guarantee contracts that preserve authenticated
security context across heterogeneous controls and requires each realized
effect to retain a current end-to-end authorization witness
\cite{continuity2026}. The protocol adopts the same general lesson that
locally valid controls do not compose automatically, but proves a different
temporal predicate. It installs an active retirement cut for one root, closes
cut-relevant historical commitments at provider-local barriers, preserves only
exact alternate-root rebinds, and balances channel tokens before issuing a
negative authorization claim: the retired atom cannot support another
protected acceptance within the manifest.

\subsection{Distributed termination and snapshots}

Dijkstra--Scholten termination detection proves passivity and message
exhaustion for a diffusing computation; Mattern develops message-counting
algorithms for asynchronous, potentially non-FIFO communication; and
Chandy--Lamport snapshots capture a consistent global state including channel
contents \cite{dijkstra1980termination,mattern1987termination,
chandy1985snapshot}. The protocol imports their treatment of in-flight
messages and cut consistency. Its target property is different: authority is
scoped to a selected root witness, protected sinks enforce an epoch fence, and
a process may remain active under an unrelated root. Durable non-process
carriers---credentials, remote jobs, triggers, and prepared effects---also
remain relevant even when all observed processes are passive.

\subsection{Access-control foundations}

Least privilege and complete mediation require authority to be both minimized
and checked at the protected operation \cite{saltzer1975protection}.
Zero-trust architecture likewise places policy enforcement near resources
rather than inferring trust from network position \cite{nist800207}.
UCONABC extends one-shot access decisions with ongoing authorization,
obligations, conditions, and mutable attributes for long-lived usage
\cite{park2004uconabc}. The protocol specializes that continuing-control
principle to retirement: it also accounts for work admitted before the final
fence, which a later access denial cannot retroactively erase. Leases provide
a complementary time-bounded right under explicit clock and failure
assumptions \cite{gray1989leases}; a lease expiry is therefore admissible only
when the registered profile proves those assumptions and covers derived work.
The protocol applies these principles to shutdown: a control-plane status
cannot substitute for a sink fence, and a credential or task is not considered
closed merely because its initiator disappeared. The resulting certificate is
an auditable authorization fact tied to a root, epoch, manifest, and set of
effect-mediating leaves.

The Authorization-Execution Gap frames a broader divergence between a
principal's intended authorization and realized execution, identifying
delegation incompleteness, channel corruption, and composition fragmentation
as structural sources and motivating authorization-integrity checks during
execution \cite{aeg2026}. The protocol instantiates a narrower protocol
question after one authorization root has been selected for retirement. It
does not infer principal intent or diagnose all three sources; it certifies the
cut, support transfer, provider obligations, and channel accounting required
for root-relative quiescence under an explicit policy and manifest.

Kubernetes finalizers and garbage collection illustrate durable, deferred
resource deletion with controller-owned cleanup keys and owner/dependent
lifecycle rules \cite{kubernetes2025finalizers,kubernetes2026garbagecollection}.
A finalizer can instantiate one provider leaf, but clearing it does not by
itself cover external effects, authorization support, or channels into other
providers.

\subsection{Positioning summary}

The closest mechanisms solve complementary slices. Cancellation requests a
state change; token revocation closes one credential; process exit observes one
runtime; graph revocation computes surviving authority; provider closure
proves a local effect frontier; evidence frameworks validate declared closure;
and termination detection closes process/message activity for one modeled
computation. The protocol composes the missing conjunction: an executed
root-epoch cut, precise multi-root support transfer, provider-local frontier
certificates, and exact cross-provider channel closure. Its certificate claims
only that conjunction and does not rename any adjacent result as global
shutdown.

\section{Assurance Boundary}
\label{sec:assurance}

The root-scoped quiescence protocol makes a strong root-relative claim: cut $c$ targets $(r,e)$, retires epoch
$e$, advances issuer state to $e+1$, and a valid certificate excludes every
future registered protected-sink acceptance whose selected support contains
the retired atom. Any cut-relevant old-atom acceptance before its provider-local
fence is bound to an exact pre-fence commitment receipt finalized after that
fence. Continuing work is permitted only when an
atomic rebind receipt selects an independently sufficient witness that excludes
the retired atom. The claim follows from the formal predicate and the deployment premises
below; it does not depend on interpreting cancellation, silence, or timeout as
success.

\subsection{Certified property and non-equivalent statements}

The certificate means \emph{root-scoped authorization quiescence}. Four nearby
statements are deliberately kept distinct:

\begin{itemize}
  \item \emph{Global idleness} would require every process and task to stop.
  Root-scoped quiescence allows shared work to continue under an independent
  root.

  \item \emph{Business completion} would require the task's objective and
  external workflow to reach a domain-specific terminal outcome. The
  certificate addresses the ability of the retired root to cause a new
  protected commit.

  \item \emph{Rollback} would undo effects committed before the applicable
  provider-local fence.
  The protocol records such effects as pre-fence committed; compensation is a
  separate authorized workflow. Classical saga execution makes compensation
  an explicit application-level action for a long-lived transaction, rather
  than an automatic consequence of abort or revocation
  \cite{garciamolina1987sagas}.

  \item \emph{Universal physical cessation} would cover uninstrumented and
  unmediated real-world consequences. The certificate covers the registered
  authority graph, endpoints, channels, and protected sinks named by its
  manifest.
\end{itemize}

These distinctions strengthen auditability: a verifier can determine exactly
which proposition was certified rather than inferring a broader notion of
``stopped.''

\subsection{Deployment premises}

\begin{table}[t]
\caption{Premises for the root-scoped no-false-quiescence result.}
\label{tab:quiescence-premises}
\small
\begin{tabularx}{\textwidth}{p{0.20\textwidth} X X}
\toprule
Premise & Required deployment fact & Result when the fact is unavailable \\
\midrule
Durable cut linearization & the root epoch transition and issuance freeze have one durable linearization point and survive restart & no positive certificate for that cut \\
Support-sound projection & every admitted live authority instance exposes its minimal sufficient root-support witnesses and selected witness & uncovered or ambiguous support is \textsc{Indeterminate} \\
No root laundering & a continuing instance changes support only through an atomic receipt selecting a current witness that excludes the retired atom & the instance remains dependent on the retired atom \\
Complete carrier coverage & credentials, tasks, descendants, triggers, reservations, prepared operations, and other registered persistence mechanisms are enumerated or soundly upper-bounded & uncovered carriers remain open or unknown \\
Sink mediation and fencing & every protected commit crosses a registered sink that rejects the retired epoch after installing the named cut fence & an unfenced or bypassable sink cannot sign a positive leaf \\
Leaf authenticity and compatibility & leaf signatures, identities, manifests, cut identifiers, configurations, and watermarks are authenticated and mutually compatible & conflicting or replayed leaves are rejected \\
Exact channel accounting & every cross-leaf authority or obligation transfer has a unique authenticated pre-cut SEND in the frozen token frontier and a matched terminal ACK; ACCEPT binds a covered destination carrier before its terminal accounting & unmatched, late-created, duplicated, causally inverted, or unbound tokens prevent global quiescence \\
Crash-safe evidence state & ledger transitions, fence state, terminal receipts, and certificate issuance are replay-safe and durably ordered & recovery resumes in nonterminal or indeterminate state \\
Lease-time validity & expiry-based closure uses an authenticated time source and the profile's stated skew and propagation bound & expiry is not accepted as closure evidence \\
\bottomrule
\end{tabularx}
\end{table}

The trusted computing base contains the durable cut ledger, certificate
verifier, cryptographic keys, registered adapter logic, authoritative state
observations, and the protected sinks that enforce fences. Agent plans, task
messages, cancellation responses, and provider data not authenticated by a
registered profile are untrusted inputs. An adapter may narrow its claim when
the provider exposes less state; it may not fill missing evidence with a
successful status.

\subsection{Threats covered by the protocol}

The safety model includes an agent that continues after cancellation, creates
descendants immediately before shutdown, replays old credentials or
certificates, duplicates requests, and attempts to preserve retired authority
through a shared principal. It includes message reordering and duplication,
controller crash and restart, delayed receipts, cancellation--completion
races, stale observations, and network partitions. These events can delay a
certificate or force \textsc{Indeterminate}; under the premises in
Table~\ref{tab:quiescence-premises}, they cannot justify a false positive.

Compromise of the certificate signing key, protected sink, authoritative
provider state, host kernel, or registered adapter violates a named premise.
Likewise, an external endpoint that offers no sound terminal query, no fence,
no bounded expiry, and no independently observable effect state is opaque for
this purpose. The correct protocol result for such a leaf is
\textsc{Indeterminate}. This is also the information-theoretic boundary: two
remote executions can present the same observations while only one retains a
late effect, so no observer restricted to those observations can distinguish
them soundly.

\subsection{Evidence-layer interpretation}

The formal proof establishes safety for all traces satisfying the model
premises. The executable state machine demonstrates that the registered
protocol and fault schedules realize those transitions. The separately
implemented trace checker validates event-chain integrity and registered invariants without
calling the controller's decision procedure. Provider or runtime adapters
establish only the profiles they implement. A successful finite evaluation
therefore supports conformance of those executions; the universal statement
comes from the conditional proof, and deployment admission comes from checking
the premises for each leaf.

Each evidence layer has a distinct role. Authenticated observations establish
provider facts, exact channel accounting establishes the global cut, registered
executions establish the tested protocol paths, the conditional proof supplies
the model-wide result, and signatures bind assertions to their issuers. The
certificate composes these layers through cut identity, manifest, configuration
digest, watermarks, and leaf digests.

\subsection{Safety, availability, and conditional completion}

No-false-quiescence is a safety property. Certificate availability is a
separate liveness property. Eventual certification requires that every
required adapter and sink recover, fences eventually install, messages are
eventually delivered or reach an authenticated terminal state, leases expire
within their admitted bounds, and fair retry makes durable progress. A
permanent partition or permanently opaque endpoint can therefore prevent a
certificate while preserving safety.

Operationally, the three outcomes have fixed meanings. \textsc{Quiescent}
authorizes reliance on the root-cut property. \textsc{Not-Quiescent} identifies
known work that must be closed or rebound. \textsc{Indeterminate} identifies an
evidence or observability gap that must be resolved or isolated. Policy may
escalate or contain either nonpositive result, but may not relabel it as
quiescence.

\section{Conclusion}
\label{sec:conclusion}

Stopping a long-running agent is an authorization problem, not a synonym for
closing its initiating process. Cancellation acknowledgement, token
revocation, and process exit each close a useful local object, yet durable
tasks, derived credentials, triggers, prepared effects, and cross-provider
messages can preserve the retiring root's ability to commit. Provider-local
closure likewise becomes a global result only when the leaves describe one
compatible cut and account for work crossing between them.

The protocol establishes that cut explicitly. A durable root-epoch
transition freezes new issuance; operation-time fences block the old epoch at
protected sinks; an outstanding authority--obligation graph tracks persistent
carriers; provider leaves certify local frontiers and watermarks; and channel
tokens close the cross-provider boundary. Minimal root-support antichains and
atomic rebind receipts preserve valid shared work without transferring the
retiring root through another principal.

The resulting certificate has a precise meaning: within its authenticated
coverage manifest and mediation premises, every cut-relevant acceptance
authorized by the retired atom and preceding its provider-local fence is
terminally accounted, and no registered protected-sink acceptance ordered
after its fence can select that atom. Known residual work is
\textsc{Not-Quiescent}; missing, conflicting, or opaque evidence is
\textsc{Indeterminate}. Neither cancellation success nor elapsed time can
silently promote either state to quiescence.

This separation gives heterogeneous agent systems a compositional shutdown
contract. Task protocols continue to manage lifecycle intent, credential
issuers continue to revoke tokens, provider adapters continue to prove local
effect closure, and distributed evidence verifiers continue to validate their
declared predicates. The protocol binds those facts to one root-scoped
authorization cut, making ``stopped'' a verifiable security assertion rather
than a control-plane impression.

\appendix
\section{Proofs of the Quiescence Results}
\label{app:proofs}

We reason over finite prefixes of accepted traces of the transition system in
Section~\ref{sec:formal}.  A rejected operation may append a diagnostic event
but does not alter roots, selected support, carrier authority, fences,
frontiers, or terminal channel dispositions.  Cryptographic and storage
claims are invoked only through the numbered assumptions in
Table~\ref{tab:quiescence-assumptions}.

\subsection{Trace notation and preservation facts}

Write
\[
 S_0\xrightarrow{o_1}S_1\xrightarrow{o_2}\cdots
 \xrightarrow{o_n}S_n
\]
for an accepted trace.  Events within the root ledger are ordered by
\(\preceq_r\); events within endpoint or sink \(p\) are ordered by
\(\preceq_p\).  Cross-domain reasoning uses only the causal happens-before
relation \(\prec\) of Section~\ref{sec:formal}.  In particular, no proof step
compares a root-ledger position with a provider position.  A transition
\emph{selects} atom \(a\) if its authenticated authorization record names a
selected witness containing \(a\).  A1 mediates issuer transitions, whereas
A5 mediates protected commitment admissions through the corresponding
provider-local barrier.  By A3, A5, and A9, root epochs, cuts, barriers, and
terminal evidence are monotone in their respective domains.

Two elementary preservation facts will be used repeatedly.  First, no
transition other than a registered authorization or rebind transition changes
\(\operatorname{sel}(x)\).  Second, a terminal acknowledgement for token
identifier \(tid\) is inserted under a uniqueness constraint binding the full
canonical token.  Replaying the same acknowledgement is idempotent; changing
its disposition is a conflict, not a second valid state.  By A1 and A7, an
old-atom token-creation record precedes the cut, and an accepted acknowledgement
precedes every terminal transition of its covered destination carrier.

\subsection{Cut and support lemmas}

\begin{proof}[Proof of Lemma~\ref{lem:cut-nonissuance}]
Let \(i\in\mathcal I_r\) be the first issuer transition after \(c\) in
\(\preceq_r\) that attempts to issue, delegate, reactivate, or enlarge
authority with a selected witness containing
\(a_c=\langle r,e,h_g\rangle\).  The atomic transaction for \(c\) retires
epoch \(e\), advances the stored root epoch to \(e+1\), and places the root in
\(\mathsf{DRAINING}\).  A3 makes those changes monotone.  A1 requires \(i\) to
validate the exact grant digest and epoch at its own root-ledger linearization
point; A2 prevents substitution of a different canonical atom.  Thus \(a_c\)
is not current for \(i\), and acceptance contradicts A1.

This argument is intentionally limited to issuer-side non-expansion.  An
already authorized carrier can retain \(a_c\) and reach a sink after \(c\)
but before that sink installs \(b_p(c)\).  Such an admission is classified in
\(\preceq_p\), must precede \(b_p(c)\), and must receive the typed receipt and
closed frontier required by A5.  An attempt after \(b_p(c)\) is rejected by
the sink fence.  Delivery or settlement following an accepted admission is
part of that same commitment event and is not a new issuer transition.
\end{proof}

\begin{proof}[Proof of Lemma~\ref{lem:support-projection}]
Consider the authenticated carrier-provenance graph in creation order.  For a
carrier \(x\) present when the cut commits, A4 requires its complete minimal support antichain and
actual selected witness to be stored with its creation or latest authorization
record.  Therefore \(a_c\in\operatorname{sel}(x)\) places \(x\) in the base
set of \(D_c\).  A different witness in \(\Sigma(x)\) cannot change this fact:
by Definition~\ref{def:support-antichain}, alternatives are distinct
conjunctive supports, while \(\operatorname{sel}(x)\) names the one that
actually authorized the current carrier.

An outstanding cut-time token whose recorded support contains \(a_c\) is an
explicit base edge even if its destination carrier has not yet materialized.
A7 either gives that token a rejecting/expiry disposition or binds acceptance
to a covered destination carrier, which the closure definition of \(D_c\)
places in the projection.

For the induction step, suppose every dependent carrier created through event
\(j\) is in the projection.  If such a carrier emits a cross-endpoint handoff,
A7 gives the handoff one canonical token and binds its selected support digest.
On acceptance, the destination record names that token and projects its
support into the new carrier.  The closure rule for \(D_c\) therefore includes
the destination carrier.  If the token is rejected or expires, no destination
carrier is created.  If its disposition is absent, the emitted token remains
an explicit unresolved edge; it is not interpreted as absence.  A6 ensures
that the endpoint and channel class appear in the frozen manifest, including
carriers materialized after the cut but before the stable-scan watermark.

Induction covers every finite creation prefix.  A4 maps missing or ambiguous
provenance to \(\mathsf{UNKNOWN}\), so the verifier reports \(U\) instead of
deleting the node.  Thus every carrier and accepted handoff selected under
\(a_c\) is included, while discarding a nonselected alternative witness cannot
discard the selected dependency.
\end{proof}

\begin{proof}[Proof of Lemma~\ref{lem:no-laundering}]
The only nonterminal transition that removes \(x\) from the dependent
projection is \(\mathsf{Rebind}(x,W',c)\).  Its guard verifies that \(W'\) is a
complete registered minimal support, that \(a_c\notin W'\), and that every atom
and derivation edge in \(W'\) is current.  A8 further requires a fresh
authorization decision and one atomic transaction that changes the selected
witness and carrier epoch while emitting \(\rho_x\).  Consequently there is
no intermediate accepted state in which the carrier continues under neither
witness, or in which it claims \(W'\) while exercising the old witness.
The same guard compares the canonical operation and effect-envelope digests
with the carrier's pre-transition record; any substitution or widening fails,
and \(\rho_x\) binds the unchanged values.

A label edit changes none of the authenticated fields in \(\rho_x\).  A parent
edit breaks its provenance digest.  Selecting an alternative that still
contains \(a_c\), contains a stale atom, or is not in \(\Sigma(x)\) fails the
guard.  Finally, descendants not covered by the atomic profile retain their
old selected records and remain in \(D_c\).  Therefore a successful departure
from \(D_c\) proves a complete current derivation independent of the retired
atom, rather than merely renaming the dependency.
\end{proof}

\subsection{Composite safety and necessary evidence}

\begin{proof}[Proof of Theorem~\ref{thm:composition-soundness}]
Assume the verifier issues \(\mathcal C_c\).  By
Definition~\ref{def:verdict}, all leaves have the same compatibility header,
the frozen manifest is complete, every required scope is covered, all
signatures and watermarks verify, every required provider barrier and frontier
is present, and exact channel conservation holds.  Each barrier authenticates
the cut before installation, so \(c\prec b_p(c)\), without any comparison of
root and provider sequence numbers.  Lemma~\ref{lem:support-projection}
therefore makes the union of the leaf projections complete for \(D_c\).

First consider the certified prefix.  Suppose a sink \(p\) accepted a
commitment \(k\) selected under \(a_c\) with
\(b_p(c)\prec_p k\).  A5 requires the local barrier to reject precisely that
admission.  Moreover, the authenticated local log and stable frontier expose
an acceptance after the barrier as a concrete \(N\) witness.  Either fact
contradicts issuance of \(Q\).  Every cut-relevant acceptance under \(a_c\)
ordered before \(b_p(c)\), including one causally after \(c\) but before barrier
installation, must instead appear in \(J_p\) with its exact
\(\rho^{\mathrm{pfc}}_{p,k}\).  Verification checks the common local domain,
operation and effect digests, strict causal order
\(k\prec_p b_p(c)\prec\rho^{\mathrm{pfc}}_{p,k}\), authenticated admission and
barrier positions, and closure of the future-use, instance, and lineage
frontier.  The leaf evidence-through boundary is later than the barrier,
receipt, and every referenced terminal acknowledgement, so a leaf cannot be
signed first and completed retrospectively.  Missing or
conflicting evidence yields \(U\), while contradictory evidence yields \(N\);
neither permits the certificate.

Now consider any admissible extension after issuance and suppose it contains
the earliest new commitment authorized by \(a_c\), at sink \(p\).  The
monotone barrier \(b_p(c)\) already exists, so the new admission is later in
\(\preceq_p\) and A5 rejects it.  The carrier cases independently reach the
same conclusion.  An effect-closed or terminal carrier cannot reach a commit
gateway.  A \(\mathsf{PRE\_FENCE\_COMMITTED}\) carrier can settle only the
exact effects named by its typed receipts because its future-use, instance,
and lineage frontier is closed.  An independently rebound carrier's selected
witness excludes \(a_c\) by Lemma~\ref{lem:no-laundering}.  Lemma~\ref{lem:cut-nonissuance}
also excludes creation or expansion of new issuer-side authority under
\(a_c\).

A cross-endpoint handoff cannot escape these cases.  Exact projected-token
equality gives every token in the frozen frontier one valid terminal
disposition, and every acceptance maps causally to a covered destination
carrier before its terminal accounting.  By A7 and
\(\operatorname{ChannelsClosed}(E)\), the sole authoritative creation event for
each old-atom token is \(\mathsf{SEND}(t)\),
\(\mathsf{SEND}(t)\prec c\), and \(t\in T_c\).  Every relevant endpoint barrier
causally follows the authenticated cut, so
\(\mathsf{SEND}(t)\prec c\prec b_p(c)\).  After the barrier, the
registered channel profile rejects new ingress selected under \(a_c\) or records a concrete
blocker.  Consequently every dependent carrier has a certified terminal,
closed, \(\mathsf{PRE\_FENCE\_COMMITTED}\), or independently rebound state,
all old-atom handoffs in \(T_c\) are accounted for, every cut-relevant old-atom pre-fence
acceptance is typed and closed, and no post-barrier admission under \(a_c\)
exists in the prefix or any extension.
These are precisely the clauses of \(\operatorname{RQ}(c)\).
\end{proof}

\begin{proof}[Proof of Theorem~\ref{thm:boundary}]
For part (1), consider two manifested endpoints \(p\) and \(q\).  Endpoint
\(p\) has closed its local carrier after emitting token \(t\), and exports
\(\mathsf{SEND}(t)\in O_p\).  Before delivery, \(q\)'s stable local inventory contains no
carrier.  Both endpoints can truthfully report that their local carrier sets
are closed, yet \(t\) can subsequently create a destination carrier if no
terminal disposition or destination fence is required.  The local statements
alone are identical to a trace in which \(t\) was rejected.  The composite
distinguishes the traces only by requiring a unique
\(\mathsf{ACK}(t,d)\) and, for acceptance, coverage of the resulting carrier.
Equivalently, composition requires equality of the canonical token projections
\(\pi_t(O(E))\) and \(\pi_t(B(E))\); it does not equate the differently typed
send and acknowledgement records.
Therefore local closure does not compose without token conservation.

For part (2), fix all visible evidence and consider an opaque reachable
endpoint \(p\).  In world \(w_0\), \(p\) has no residual root-dependent
carrier.  In world \(w_1\), it has one delayed callback holding \(a_c\) and
able to commit a protected effect.  By hypothesis, \(p\) exposes neither an
inventory capable of distinguishing these states nor a fence/frontier that
eliminates the callback's effect path.  Every observation available to a
verifier outside \(p\) is therefore the same in \(w_0\) and \(w_1\).

Any deterministic or randomized verifier receiving the same observation has
the same output distribution in the two worlds.  If it returns \(Q\), it is
unsound in \(w_1\); if it never returns \(Q\), it is incomplete in \(w_0\).
No observer of the remaining evidence is both sound and complete.  Recording
the missing endpoint as \(U\) preserves soundness and identifies the exact
evidence or enforcement capability needed to resolve the verdict.
\end{proof}

\subsection{Preservation, algebra, and recovery}

\begin{proof}[Proof of Theorem~\ref{thm:independent-support-preservation}]
By premise, \(W'\) is a registered sufficient witness, excludes \(a_c\), and
is current.  Definition~\ref{def:rebind} re-runs authorization under \(W'\)
and atomically installs it as the selected witness.  The resulting carrier
therefore has exercisable authority derived from every atom in \(W'\), not
from \(a_c\).  Retiring \(a_c\) changes neither the epoch nor the grant digest
of an atom in \(W'\), so the cut cannot invalidate that derivation.

The rebind receipt proves the support change to the leaf verifier, permitting
the carrier to leave \(D_c\) without being terminated.  The composite
certificate names \(a_c\) and is evaluated only against dependencies selected
under that atom; it contains no transition revoking an atom in \(W'\).  Thus
the protocol removes exactly the retired root's contribution while preserving
independently supported work.  A later change to an atom in \(W'\) is handled
by its own epoch transition and is outside this certificate's effect.
\end{proof}

\begin{proof}[Proof of Theorem~\ref{thm:merge-algebra}]
Fix compatibility header \(\chi\).  For each leaf identifier, the merge rule
is the scalar operation
\[
 \bot\sqcup z=z,\qquad z\sqcup z=z,\qquad
 z\sqcup z'=\mathsf{conflict}\quad(z\ne z'),
\]
with \(\mathsf{conflict}\sqcup z=\mathsf{conflict}\).  This is the join on a
flat domain whose bottom is absence and whose top is authenticated conflict.
Its case definition is associative, commutative, and idempotent.  Finite
pointwise products of such joins retain all three properties.  Header
disagreement maps to the same absorbing conflict element, so it does not
change the argument.

The canonical composite serialization sorts leaf identifiers and token keys;
therefore equal maps have equal composite digests independent of message
arrival order.  Definition~\ref{def:verdict} is a pure function of that map,
its verified facts, and the frozen manifest.  Duplicate delivery changes none
of them.  Different authenticated bodies under one leaf identifier remain at
\(\mathsf{conflict}\), which triggers \(U\) and cannot be overwritten by a
later duplicate.  The digest and verdict are consequently merge-order
independent.
\end{proof}

\begin{proof}[Proof of Theorem~\ref{thm:crash-replay}]
Partition a crash around each security-changing transaction's durable
linearization point.  Before the point, recovery observes no transition and a
retry may execute normally.  After the point, recovery observes the complete
transition and the same idempotency key returns its committed result.  A3 and
A9 make epoch advances, cut records, fences, terminal carrier states, closed
frontiers, and terminal token dispositions monotone.  A retry proposing a
different body under the same key is authenticated conflict, never a
replacement.

An old grant request names retired epoch \(e\), so every later issuer
transition fails after the cut advances \(r\) to \(e+1\).  A replayed carrier
message is handled separately at its manifested endpoint: before the local
barrier it remains an explicitly accounted transfer or acceptance, and after
the barrier it is rejected.  A replayed fence or close record is idempotent.  A
second token disposition violates the unique canonical token binding and
becomes a conflict.  A rebind receipt is bound to \(cid\), carrier id, old and
new support digests, exact operation and effect-envelope digests, profile
digest, and a single-use transition key; replay can neither rebind another
carrier, widen its effect envelope, nor create a second accepted transition.

For certificate issuance, the authenticated certificate binds \(qid\) and
\(h_q\), and the durable issuance record commits
\((cid,qid,h_q,H(E),s_c,\mathcal C_c)\), where \(h_q\) digests the canonical
pre-certificate body.
Repeating that exact identifier and candidate digest finds
the durable record and returns the identical signature and certificate
identity.  Reusing \(qid\) with a different body is an idempotency conflict;
using another request identifier for the already completed \(cid\) is a
conflicting remint.  A9 rejects both cases.  Certificate replay therefore
states the same historical fact and cannot upgrade different evidence or mint
a second certificate for the generation.

For locality, let root identifier \(r'\ne r\) be absent from the frozen
manifest dependencies and from \(\operatorname{rid}(W)\) for every selected
witness \(W\) represented by \(E\).  Advancing \(r'\)'s epoch or changing its
ledger phase alters none of \(\chi\), the leaf bodies,
channel reconciliation, rebind receipts, or frontier facts on which
\(V_c(E)\) depends.  The canonical projection used to verify
\(\mathcal C_c\) is unchanged.  Hence unrelated-root advancement cannot
invalidate the root-relative certificate, while any shared selected witness
would, by definition, make the root related and require fresh evidence.
\end{proof}

\subsection{Conditional convergence}

\begin{proof}[Proof of Theorem~\ref{thm:conditional-liveness}]
Lemma~\ref{lem:cut-nonissuance} prevents issuer-side creation or expansion of
authority selected under \(a_c\) after \(c\).  A10 and fair execution install
every manifested \(b_p(c)\); A5 then prevents further commitment admissions
under \(a_c\) at each sink.  Pre-cut tokens delivered or accepted while barriers
are propagating remain in \(D_c\) through their exact channel records rather
than escaping the population.  A11 directly supplies finiteness of the complete
cut-dependent population of carriers, effects, and tokens, and ensures that
each emitted token is eventually delivered to a covered destination or obtains
an authenticated fence-reject or enforced-expiry disposition.

After the last required barrier is durable, consider the finite set of
unresolved cut obligations.  It contains one obligation for every dependent
carrier not yet closed, terminal, \(\mathsf{PRE\_FENCE\_COMMITTED}\), or
rebound; every token without a valid terminal disposition; every missing
fence/frontier attestation; every transient pre-certification unknown or
resolvable state disagreement; and every missing compatible leaf scope.  Token acceptance transfers
the same causal obligation from in transit to a covered destination carrier,
so it does not remove the obligation from this set prematurely.  Under A7,
A10, A11, and the theorem premises, fair execution eventually discharges each
member.  Monotonicity prevents a discharged terminal fact from reopening, and
A10 clears transient pre-certification unknowns and state disagreements.  An
authenticated same-key equivocation is absorbing for the generation and would
permanently preclude \(Q\); the theorem's no-equivocation premise excludes that
case rather than treating it as clearable evidence.

The resulting compatible evidence covers every manifest scope and every
provider frontier.  Each acceptance in \(K_c^p\) precedes its local barrier and
has an exact typed receipt, every continuing carrier has an admissible state,
and
\(\pi_t(O(E))=\pi_t(B(E))\) with each acceptance linked to a covered carrier.
A5 excludes an accepted post-barrier commitment, so there is no \(N\) witness;
A10 leaves no \(U\) fact.  Definition~\ref{def:verdict} returns \(Q\), and the
verifier issues \(\mathcal C_c\).  Before all obligations discharge, a
concrete open element produces \(N\) or a missing, opaque, or conflicting fact
produces \(U\); neither authorizes a certificate.
\end{proof}

\bibliographystyle{ACM-Reference-Format}
\setlength{\bibsep}{-0.2pt}
\bibliography{references}

\end{document}